%% file: main.tex
\documentclass[11pt]{article}

\input{macros}

\begin{document}

\title{A Quantum Unique Games Conjecture}

\author[1]{Hamoon Mousavi}
\author[2]{Taro Spirig}

\affil[1]{Simons Institute for Theoretical Computer Science, University of California, Berkeley\\
hmousavi@berkeley.edu}
\affil[2]{QMATH, Department of Mathematical Sciences, University of Copenhagen\\
tasp@math.ku.dk}

\maketitle
    
\begingroup
        \makeatletter
    \renewcommand{\@makefnmark}{}
    \makeatother
    
    \renewcommand\thefootnote{}
    \footnotetext{This material is based upon work supported by the U.S. Department of Energy, Office of Science, National Quantum Information Science Research Centers, Quantum Systems Accelerator. TS is supported by the European Union under the Grant Agreement No 101078107, QInteract and VILLUM FONDEN via Villum Young Investigator grant (No 37532) and the QMATH Centre of Excellence (Grant No 10059).}
\endgroup

\begin{abstract}
After the NP-hardness of computational problems such as $3$SAT and MaxCut was established, a natural next step was to explore whether these problems remain hard to approximate. While the quantum extensions of some of these problems are known to be hard—indeed undecidable—their inapproximability remains largely unresolved. In this work, we introduce definitions for the quantum extensions of Label-Cover and Unique-Label-Cover. We show that these problems play a similarly crucial role in studying the inapproximability of quantum constraint satisfaction problems as they do in the classical setting.
\end{abstract}

\tableofcontents

\section{Introduction}\label{sec:introduction}
\subsection{Motivations}\label{sec:motivations}
The Label-Cover problem~\cite{ben-or,feige-kilian,arora0}, a well-known constraint satisfaction problem (CSP), has been central in the study of hardness of approximation. According to the PCP theorem~\cite{ben-or,babai,babai2,fortnow,feige,arora1,arora2,raz}, it is NP-hard to approximate Label-Cover within any additive constant. This result, combined with gap-preserving reductions from Label-Cover, has led to optimal inapproximability results for several other key problems, including $3$XOR and $3$SAT, due largely to the work of H\r{a}stad~\cite{hastad1,hastad2}.

An instance of Label-Cover consists of a bipartite graph with left vertex set $U$ and right vertex set $V$, a large label set (also called an alphabet) $\mathcal{L}$, and maps $\pi_{u,v}:\mathcal{L} \to \mathcal{L}$ associated with each edge $(u,v)$. These maps are usually called projections. A \emph{labelling} $a:U\cup V \to \mathcal{L}$ is said to satisfy an edge $(u,v)$ if it respects the associated projection map, meaning if $\pi_{u,v}(a(u)) = a(v)$. The \emph{value of the instance} is the maximum fraction of satisfied edges over all possible labelings. 

In quantum information theory, the quantum extension of Label-Cover, known as \emph{entangled projection games}~\cite{parallelreptition-Vidick}, is a special case of \emph{nonlocal games}~\cite{feige-lovasz,cleve_consequences_and_limits}.\footnote{We primarily adopt the CSP terminology, while referencing certain naming conventions from the nonlocal games literature for the convenience of readers familiar with that field. See Section \ref{sec:intro-quantum-value} for a discussion of our terminology.} The primary distinction lies in the definition of labeling. A \emph{quantum labeling} generalizes probabilistic labeling. In classical probabilistic labeling, each vertex is assigned a probability distribution over labels. In quantum labeling, each edge is assigned a distribution over pairs of labels—one for each vertex incident on the edge. 

More formally, a quantum labeling assigns quantum measurements to each vertex, with the measurement outcomes corresponding to labels. When both vertices of an edge are measured, the resulting probability distribution determines the likelihood that the edge is satisfied. The \emph{quantum value} of the instance is the maximum expected fraction of satisfied edges over all possible quantum labelings. A formal definition is given in Section \ref{sec:general-csps}.

In the quantum setting, the role of the PCP theorem is played by the $\MIP^* = \RE$ theorem of Ji et al.~\cite{ji_mip_re}: it is $\RE$-hard to approximate the quantum value of Label-Cover within any additive constant.\footnote{Recall that $\RE$ is the class of recursively enumerable languages. Stating that a problem is $\RE$-hard is another way of saying that it is undecidable.} Now, as with the classical case, one might expect to derive optimal inapproximability results for other quantum constraint satisfaction problems. O’Donnell and Yuen~\cite{odonnell-yuen}, for instance, extended H\r{a}stad’s results to the optimal inapproximability of quantum $3$XOR. So far, the classical and quantum theories of inapproximability have followed parallel trajectories.

In the classical development, after the initial success of H\r{a}stad, researchers discovered some of the limitations of Label-Cover. To address these, Khot~\cite{khot_original} introduced a modification to the Label-Cover problem, imposing a uniqueness condition on the projection maps. This modified problem is known as the Unique-Label-Cover and it remains an open question whether it is hard to approximate. However, assuming that Unique-Label-Cover is hard to approximate, researchers showed optimal inapproximability results for a variety of other problems, including MaxCut~\cite{kkmo}. This assumption is known as the Unique Games Conjecture (UGC)~\cite{khot_original}. 

A striking divergence arises between the classical and quantum theories of inapproximability with the introduction of Unique-Label-Cover. Kempe et al.~\cite{kempe} showed that the quantum value of Unique-Label-Cover is efficiently approximable. This implies that a \emph{direct} quantum analogue of the UGC does not hold. As a result, none of the classical consequences of the UGC can be straightforwardly extended to the quantum domain. This presents a significant challenge to the development of a comprehensive quantum theory for hardness of approximation. 

For instance, understanding the complexity of the anti-ferromagnetic Heisenberg model—a quantum generalization of MaxCut and a problem of great significance in physics and quantum information—remains an open question. While this problem is known to be $\QMA$-hard, the complexity of approximating it is still unresolved.\footnote{Recall that $\QMA$ is the quantum analogue of $\NP$.} Hwang et al.~\cite{quantum-max-cut-integrality-gap} made significant progress on this question, but their results rely on the classical UGC, which limits them to showing $\NP$-hardness rather than $\QMA$-hardness.

\begin{center}\emph{Finding the right candidate for a quantum extension of UGC is key for advancing the theory of inapproximability in the quantum setting. This work proposes such a candidate.}
\end{center}

We begin by noting that Kempe et al.’s algorithm for approximating the quantum value of Unique-Label-Cover applies only to some of many possible quantum values that are defined in the literature on nonlocal games. Examples of quantum values include, but are not limited to, the tensor-product, commuting-operator, synchronous, and oracularizable-synchronous values. These arise from different models for quantum entanglement. Label-Cover remains hard across all these models~\cite{slofstra_tsirelsons_problem_and_an_embedding_theorem,slofstra_set_of_quantum_correlations,ji_mip_re}, whereas Unique-Label-Cover is efficiently approximable in only some of them~\cite{kempe}. 

We choose a model where Unique-Label-Cover might still be hard to approximate. This choice forms the basis of our \emph{quantum Unique Games Conjecture} (qUGC). We give an informal overview in Section \ref{sec:intro-quantum-unique-games-conjecture}. Formal definitions can be found in Section~\ref{sec:label-cover}.

We then show how classical reductions from Unique-Label-Cover can be ``quantized'' to prove inapproximability results for quantum CSPs. We illustrate this approach with $2$-Lin~\cite{khot_original} and MaxCut~\cite{kkmo}, and we believe the same tools can be applied to Raghavendra’s seminal work on CSPs~\cite{prasad}. An informal overview is given in Section~\ref{sec:result-on-quantum-maxcut}. The proof ideas are sketched in Section~\ref{sec:proof-ideas}. Detailed statements and proofs are provided in Sections~\ref{sec:2lin} and~\ref{sec:maxcut}. 

We conclude this section with a few remarks.

\textbf{Relationship between UGC and qUGC.} The cases of $3$XOR, $2$-Lin, and MaxCut could be instances of a broader principle: any reduction from Label-Cover (or its variants such as Unique-Label-Cover, Smooth-Label-Cover, and $2$-to-$2$ Label-Cover) might be adaptable to the quantum setting. If true, this would imply that the classical UGC implies qUGC. We explore this possibility and related open questions in Section~\ref{sec:future}.

\textbf{Variants of qUGC.} Sections~\ref{sec:intro-quantum-unique-games-conjecture} and~\ref{sec:label-cover} present several variants of the quantum Unique Games Conjecture. These variants may help in proving hardness for other models, such as the anti-ferromagnetic Heisenberg model or CSPs investigated in~\cite{original}. These topics are further detailed as Problems 4, 7, and 8 in Section~\ref{sec:future}.

\textbf{Alternative Approaches to UGC.} Some of the variants of quantum UGC proposed in this paper may be easier to prove hardness for than the classical UGC itself. Proving hardness for any of these variants could strengthen confidence in the UGC. Conversely, some of our UGC variants may be easier to disprove, and doing so could offer insights into disproving UGC. Figure \ref{fig:unique-label-cover-final-complexity-landscape}, illustrates the reductions between variants of Unique-Label-Cover. 

\textbf{Universality.} Label-Cover is exhibiting a remarkable \emph{universality}, remaining hard across all the different quantum and classical models, see Figure \ref{fig:label-cover-complexity-landscape}. In contrast, the brittleness of Unique-Label-Cover aptly mirrors the difficulty in resolving the UGC. From another perspective, however, the constraints encountered in the search for a quantum analogue of the UGC could serve as a guiding principle in identifying the most appropriate model for quantum value from a computational standpoint. 

This paper builds on our previous work on noncommutative CSPs~\cite{original}, exploring how noncommutativity reshapes the landscape of CSPs. One of our motivations is the direct applications of this research to quantum information. We are also equally interested in whether a broader perspective on CSPs might uncover unifying principles that could provide new insights into classical CSPs. Our remarks on alternative approaches to UGC and universality are examples of the insights we are pursuing.

\subsection{Quantum Value}\label{sec:intro-quantum-value}
We illustrate the concept of \emph{quantum value} of a CSP instance using the example of MaxCut. For formal definitions refer to Section \ref{sec:csps}. Before proceeding, we briefly outline the connection to \emph{nonlocal games}; however, no prior knowledge in quantum information is required to follow this work. 

\textbf{Nonlocal Games.} The types of values we introduce for CSPs correspond directly to the values for \emph{nonlocal games} previously studied in quantum information. Every CSP instance can be viewed as a nonlocal game and vice versa and the correspondence between their values is straightforward. This relationship is detailed in Section 10.2 of~\cite{original}. The physical, mathematical, and computational considerations for the definitions of the various types of values are extensively explored in the literature on nonlocal games~\cite{bell,chsh,tsirelson1,tsirelson2,scholz_tsirelsons_problem,scholz2008tsirelson,ozawa_connes_embedding,junge_connes_embedding_problem,fritz_tsirelsons_problem_and_kirchbergs_conjecture,cleve_consequences_and_limits,paulsen2016estimating,helton2017algebras,kim2018synchronous,slofstra_set_of_quantum_correlations,slofstra_tsirelsons_problem_and_an_embedding_theorem,ji_mip_re,pi2}. For readers familiar with this area, Table \ref{tab:csp-two-provers} summarizes the correspondences relevant to this paper. Please also refer to the discussion on previous work at the end of this section for our choice of terminology.

\begin{table}[h!]
\centering
\begin{tabular}{ c|c } 
 CSP terminology & Nonlocal game terminology \\ 
 \hline
 \hline
 classical (Definition \ref{def:classical-value}) & classical synchronous \\ 
 \hline
 noncommutative (Definition \ref{def:noncommutative-value} and~\cite{original}) & quantum synchronous~\cite{paulsen2016estimating,helton2017algebras,kim2018synchronous,pi2} \\ 
 \hline
 quantum (Definition \ref{def:quantum-value}) & oracularizable quantum synchronous~\cite{helton2017algebras,ji_mip_re,pi2} \\
\end{tabular}
\caption{Correspondence between types of values for CSPs and nonlocal games.}
\label{tab:csp-two-provers}
\end{table}

\noindent\textbf{MaxCut.} The objective of MaxCut is to partition the vertices of a given graph $G=(V,E)$ into two subsets, such that the number of edges crossing the partition is maximized (as illustrated in Figure \ref{fig:maxcut-illustration}). 
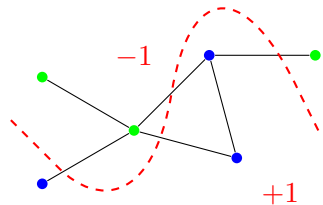
\begin{figure}[ht]
    \centering
    \begin{tikzpicture}
        \node[circle,fill=green, inner sep=0.05cm] (1) at (0,0) {};
        \node[circle,fill=blue, inner sep=0.05cm] (2) at (1,1) {};
        \node[circle,fill=blue, inner sep=0.05cm] (3) at (1.366,-0.366) {};
        \node[circle,fill=green, inner sep=0.05cm] (4) at (2.41,1) {};
        \node[circle,fill=green, inner sep=0.05cm] (5) at (-1.22,0.71) {};
        \node[circle,fill=blue, inner sep=0.05cm] (6) at (-1.22,-0.71) {};
        \draw[] (1) -- (2);
        \draw[] (1) -- (3);
        \draw[] (2) -- (3);
        \draw[] (2) -- (4);
        \draw[] (1) -- (5);
        \draw[] (1) -- (6);
        \draw[thick, red, dashed] (2.45,0) -- (2.2,0.5) .. controls (1.5,2) and (0.75,2)  .. (0.5,0.5) .. controls (0.25,-1) and (-0.5,-1) .. (-1,-0.5) -- (-1.7,0.2);
        
        \node at (1) [above=0.7cm,red]{$-1$};
        \node at (3) [below right=0.2cm,red]{$+1$};
    \end{tikzpicture}
    
    \caption{An instance of Max-Cut and an example of a partition (cut).}
    \label{fig:maxcut-illustration}
\end{figure}
\,\\

\noindent\textbf{Classical Value.} A (classical) partition can be represented as an assignment $x:V\to \{-1,+1\}$. The value of this partition, meaning the number of edges that cross the partition, is given by $\sum_{(i,j) \in E} \frac{1-x_ix_j}{2}$. The (classical) value of the instance, denoted $\omega_c(G)$, is the maximum partition value over all possible (classical) partitions:
\begin{equation}\label{eq:classical-value-max-cut}
    \openup\jot
    \begin{aligned}[t]
    \text{ maximize:}\quad &\sum_{(i,j) \in E} \frac{1-x_ix_j}{2} \\
            \text{subject to:}\quad & x_i \in \{-1,+1\}.
    \end{aligned}
\end{equation}
\,\\

\noindent\textbf{Noncommutative Value.} A \emph{noncommutative partition} assigns an observable to each vertex, i.e.\ $X:V \to \Obs(\mathcal{H})$, where $\mathcal{H}$ is some finite-dimensional Hilbert space. Observables are unitary operators with eigenvalues in the set $\{-1,+1\}$ (refer to Section \ref{sec:measurements} for background on quantum measurements and observables). If we substitute observables for $\pm 1$ in \eqref{eq:classical-value-max-cut}, we obtain the noncommutative value of the instance, denoted $\omega_\nc(G)$:
\begin{equation}\label{eq:noncommutative-value-max-cut}
\openup\jot
\begin{aligned}[t]
\text{ maximize:}\quad &\sum_{(i,j)\in E} \frac{1-\tr(X_iX_j)}{2} \\
        \text{subject to:}\quad & X_i \in \Obs(\mathcal{H}),
\end{aligned}
\end{equation}
where $\tr$ is the dimension-normalized trace on $\mathcal{H}$. Crucially, the optimization is over all finite-dimensional $\mathcal{H}$. When $\mathcal{H}$ is one-dimensional, we recover the classical value, therefore $\omega_\nc(G) \geq \omega_c(G)$. 

The noncommutative value has been the focus of our previous work on CSPs~\cite{original}. The main focus of this paper, however, is the \emph{quantum value} of CSPs, which we define next.
\,\\

\noindent\textbf{Quantum Value.} A \emph{quantum partition} lies somewhere in between a classical and noncommutative partition: A noncommutative partition $X$ is \emph{quantum} if $X_i$ and $X_j$ commute whenever $(i,j)$ is an edge. In quantum mechanics, commuting observables can be simultaneously measured, emphasizing the physical motivation for this additional commutation relation. 

The \emph{quantum value}, denoted $\omega_q(G)$, is therefore obtained from \eqref{eq:noncommutative-value-max-cut} by imposing the extra commutation relations: 
\begin{equation}\label{eq:quantum-value-max-cut}
\openup\jot
\begin{aligned}[t]
\text{ maximize:}\quad &\sum_{(i,j)\in E} \frac{1-\tr(X_iX_j)}{2} \\
        \text{subject to:}\quad & X_i \in \Obs(\mathcal{H}),\\
                          \quad & X_iX_j=X_jX_i \text{ for all } (i,j) \in E.
\end{aligned}
\end{equation}
Since every classical assignment is a quantum assignment and every quantum assignment is a noncommutative assignment, we obtain the chain of inequalities $\omega_\nc(G) \geq \omega_q(G) \geq \omega_c(G)$. 

Using the same approach outlined here, we can define the classical and quantum values for all CSPs (refer to Section \ref{sec:csps} for the relevant definitions). However, the noncommutative value is defined only for CSPs where each constraint involves only two variables.\footnote{This restriction is related to a similar issue for synchronous value of nonlocal games: with current definitions in the literature, the quantum synchronous value of a three-player nonlocal game reduces to the classical synchronous value. See also the remark after Definition \ref{def:noncommutative-value}.} 

We assume $\omega_c,\omega_q,\omega_\nc$ are normalized appropriately so that they are between $0$ and $1$. So for example in \eqref{eq:classical-value-max-cut},\eqref{eq:noncommutative-value-max-cut}, and \eqref{eq:quantum-value-max-cut}, the normalization amounts to dividing the objective function by the number of edges.

\textbf{Remark.} It is important to note that the commutativity requirement in a quantum assignment does not imply that all observables in the assignment commute with each other, nor does it mean that all observables can be diagonalized in the same basis. Commutativity is not a transitive property: every operator commutes with the identity operator, but not every pair of operators commute with each other. For example, in MaxCut, if two vertices are at a distance of two or more, their observables in a quantum assignment need not commute.

\,

\noindent\textbf{Previous Work.} To the best of our knowledge, the first article introducing a concept closely related to quantum assignment, under the name \emph{oracularizable strategies} and in the context of entangled nonlocal games, is by Ito et al.~\cite{ito-oracularization}. Oracularization is a widely used technique in the study of classical multiprover interactive proofs (see the references in \cite{ito-oracularization}). Helton et al.~\cite{helton2017algebras} approached \emph{oracularizable strategies} from a mathematical standpoint, using the term \emph{locally commuting algebras} (see Sections 7 and 8 of their paper). 

Our definition of quantum assignment is the direct translation (for CSPs) of the \emph{oracularizable quantum synchronous strategies} in~\cite{pi2}. The definition in~\cite{pi2} itself dates back to the breakthrough results~\cite{natarajan_neexp,ji_mip_re}. In this work, we refer to these assignments as \emph{quantum} because the requirement of \emph{commutativity} stems from \emph{simultaneous measurability in quantum mechanics}. In contrast, we refer to assignments as \emph{noncommutative} when there is no commutation requirement. This terminology originates from our earlier work on noncommutative CSPs~\cite{original}.\footnote{Our decision to use \emph{quantum synchronous strategies}, both here and in~\cite{original}, is driven by the syntactical simplicity they provide compared to tensor-product strategies. While we believe that all of our results can be extended to the tensor-product model, we leave a more detailed exploration of this for future work.}

\subsection{Quantum Unique Games Conjecture}\label{sec:intro-quantum-unique-games-conjecture}
Before moving on to our conjecture, we first briefly review decision problems and reductions between them. Readers familiar with the terminology can skip the next section. 

\subsubsection{Reductions} 
In the context of hardness of approximation, it is often easier to work with a version of a computational problem known as the \emph{decision problem with a promise}. Let $\P$ be any computational problem for which $\omega_c,\omega_q,\omega_\nc$ are defined. For example $\P$ could be MaxCut, Label-Cover, or Unique-Label-Cover. For every $s,t \in \{c,q,\nc\}$ and $\zeta,\delta > 0$, the \emph{decision problem} $\P_{s,t}(1-\zeta,\delta)$ is defined as follows: given an instance $\phi\in \P$ that is promised to be either
\begin{itemize}
    \item $\omega_s(\phi) \geq 1 - \zeta$, called a ``yes'' instance, or
    \item $\omega_t(\phi) < \delta$, called a ``no'' instance,
\end{itemize} 
decide whether the instance is a ``yes'' or ``no'' instance. When $s=t$, we simply write $\P_{s}(1-\zeta,\delta)$ to denote $\P_{s,s}(1-\zeta,\delta)$.

A \emph{reduction} from $\P_{s,t}(\gamma,\delta)$ to $\P'_{s',t'}(\gamma',\delta')$ is any efficient map $\P \to \P'$ that is
\begin{itemize}
\item complete: it maps ``yes'' instances to ``yes'' instances, and
\item sound: it maps ``no'' instances to ``no'' instances.
\end{itemize}
As a simple example, for every choice of $\zeta,\delta > 0$, the decision problem $\P_{q,\nc}(1-\zeta,\delta)$ reduces to $\P_{\nc,\nc}(1-\zeta,\delta)$ via the identity map. Soundness is trivial, and completeness holds because $\omega_\nc(\phi) \geq \omega_q(\phi)$ for all instances $\phi$. This is denoted by $\P_{q,\nc} \longrightarrow \P_{\nc,\nc}$ (we often drop the completeness and soundness parameters in our informal discussions in this introduction). There are several other trivial reductions between decision versions of $\P$, and they are illustrated in Figure \ref{fig:trivial-reductions}. 
\begin{figure}[ht]
\centering
	\begin{tikzpicture}

        \node[] (c_nc) at (0,4) {$\P_{c,\nc}$};   \node[] (c_q) at (2,4) {$\P_{c,q}$};    \node[] (c_c) at (4,4) {$\P_{c,c}$};
        \node[] (q_nc) at (0,2) {$\P_{q,\nc}$};   \node[] (q_q) at (2,2) {$\P_{q,q}$};    \node[] (q_c) at (4,2) {$\P_{q,c}$};
        \node[] (nc_nc) at (0,0) {$\P_{\nc,\nc}$};\node[] (nc_q) at (2,0) {$\P_{\nc,q}$}; \node[] (nc_c) at (4,0) {$\P_{\nc,c}$};
        
        \draw[ ->] (c_nc) to (c_q);\draw[ ->] (c_q) to (c_c);
        \draw[ ->] (q_nc) to (q_q);\draw[ ->] (q_q) to (q_c);
        \draw[ ->] (nc_nc) to (nc_q);\draw[ ->] (nc_q) to (nc_c);

        \draw[ ->] (c_nc) to (q_nc);\draw[ ->] (c_q) to (q_q);\draw[ ->] (c_c) to (q_c);
        \draw[ ->] (q_nc) to (nc_nc);\draw[ ->] (q_q) to (nc_q);\draw[ -> ] (q_c) to (nc_c);
    \end{tikzpicture}
    \caption{Trivial Reductions: if a problem in this diagram is hard (for some complexity class), then every problem reachable from it is at least as hard.}
\label{fig:trivial-reductions}
\end{figure}
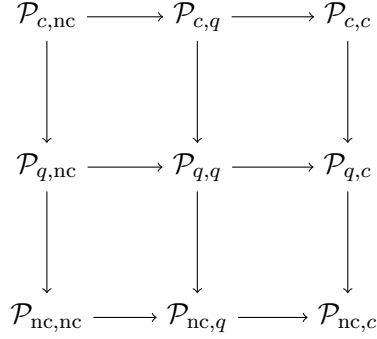 

Lastly, we use the abbreviations LC for Label-Cover and ULC for Unique-Label-Cover.

\subsubsection{Label-Cover}
The PCP theorem states that the classical value of Label-Cover is $\NP$-hard to approximate, or more precisely:
\begin{theorem*}[PCP Theorem] For every $\delta> 0$, there is a large enough alphabet over which $\LC_{c}(1,\delta)$ is $\NP$-hard.
\end{theorem*}
Since we have the trivial reductions $\LC_{c,c} \longrightarrow \LC_{q,c} \longrightarrow \LC_{\nc,c}$, we conclude that both $\LC_{q,c}$ and $\LC_{\nc,c}$ are at least $\NP$-hard. This is illustrated in Figure \ref{fig:label-cover-pcp-implication}.
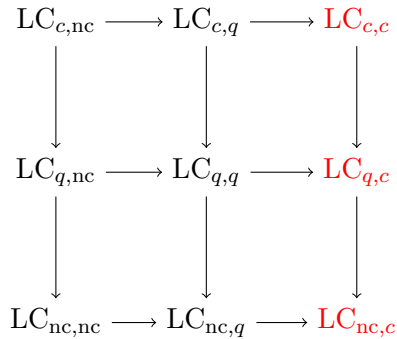
\begin{figure}[ht]
\centering
	\begin{tikzpicture}

        \node[] (c_nc) at (0,4) {$\LC_{c,\nc}$};   \node[] (c_q) at (2,4) {$\LC_{c,q}$};    \node[red] (c_c) at (4,4) {$\LC_{c,c}$};
        \node[] (q_nc) at (0,2) {$\LC_{q,\nc}$};   \node[] (q_q) at (2,2) {$\LC_{q,q}$};    \node[red] (q_c) at (4,2) {$\LC_{q,c}$};
        \node[] (nc_nc) at (0,0) {$\LC_{\nc,\nc}$};\node[] (nc_q) at (2,0) {$\LC_{\nc,q}$}; \node[red] (nc_c) at (4,0) {$\LC_{\nc,c}$};
        
        \draw[ ->] (c_nc) to (c_q);\draw[ ->] (c_q) to (c_c);
        \draw[ ->] (q_nc) to (q_q);\draw[ ->] (q_q) to (q_c);
        \draw[ ->] (nc_nc) to (nc_q);\draw[ ->] (nc_q) to (nc_c);

        \draw[ ->] (c_nc) to (q_nc);\draw[ ->] (c_q) to (q_q);\draw[ ->] (c_c) to (q_c);
        \draw[ ->] (q_nc) to (nc_nc);\draw[ ->] (q_q) to (nc_q);\draw[ -> ] (q_c) to (nc_c);
    \end{tikzpicture}
    \caption{Implication of the PCP theorem: problems in the third column are $\NP$-hard.}
\label{fig:label-cover-pcp-implication}
\end{figure}

What can we say about the other variants of Label-Cover in Figure \ref{fig:label-cover-pcp-implication}? The noncommutative analogue of the PCP theorem is the $\MIP^*=\RE$ theorem. This theorem states that the noncommutative value of Label-Cover is $\RE$-hard to approximate. In fact~\cite{ji_mip_re} proves a stronger theorem:
\begin{theorem*}[$\MIP^*=\RE$, informal] For every $\delta> 0$, the decision problem $\LC_{q,\nc}(1,\delta)$ is $\RE$-hard.
\end{theorem*}
This immediately implies that the problems in the last two rows in Figure \ref{fig:label-cover-pcp-implication} are $\RE$-hard. Our full knowledge of the complexity landscape of Label-Cover is summarized in Figure \ref{fig:label-cover-complexity-landscape}.
\begin{figure}[ht]
\centering
	\begin{tikzpicture}

        \node[] (c_nc) at (0,4) {$\LC_{c,\nc}$};   \node[] (c_q) at (2,4) {$\LC_{c,q}$};    \node[] (c_c) at (4,4) {$\LC_{c,c}$};
        \node[] (q_nc) at (0,2) {$\LC_{q,\nc}$};   \node[] (q_q) at (2,2) {$\LC_{q,q}$};    \node[] (q_c) at (4,2) {$\LC_{q,c}$};
        \node[] (nc_nc) at (0,0) {$\LC_{\nc,\nc}$};\node[] (nc_q) at (2,0) {$\LC_{\nc,q}$}; \node[] (nc_c) at (4,0) {$\LC_{\nc,c}$};

        \draw[dotted,thick] (-2.5,3) -- (6.8,3);
        \draw[dotted,thick] (3,5.5) -- (3,3);
        \node[] at (6,5) {\bf NP-hard};
        \node[] at (-2,5) {\bf NP};
        \node[] at (6,1) {\bf RE-hard};
        
        \draw[ ->] (c_nc) to (c_q);\draw[ ->] (c_q) to (c_c);
        \draw[ ->] (q_nc) to (q_q);\draw[ ->] (q_q) to (q_c);
        \draw[ ->] (nc_nc) to (nc_q);\draw[ ->] (nc_q) to (nc_c);

        \draw[ ->] (c_nc) to (q_nc);\draw[ ->] (c_q) to (q_q);\draw[ ->] (c_c) to (q_c);
        \draw[ ->] (q_nc) to (nc_nc);\draw[ ->] (q_q) to (nc_q);\draw[ -> ] (q_c) to (nc_c);
    \end{tikzpicture}
    \caption{The complexity landscape of Label-Cover.}
\label{fig:label-cover-complexity-landscape}
\end{figure}
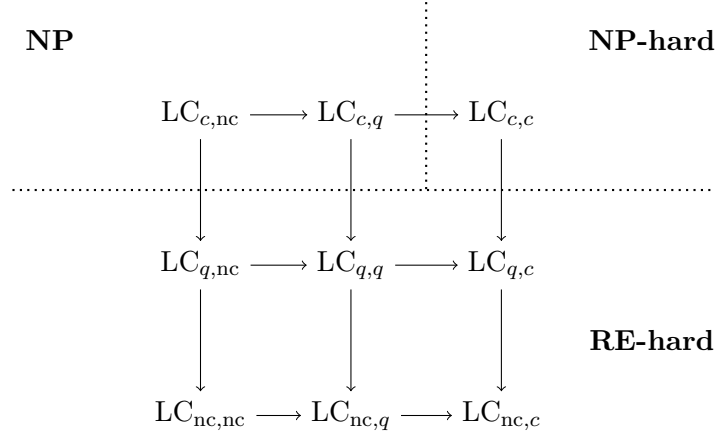 

\subsubsection{Unique-Label-Cover}
A Unique-Label-Cover instance is a special case of Label-Cover instance where every projection map is a bijection (i.e.\ a permutation), see Section \ref{sec:label-cover} for further details. The well-known Unique Games Conjecture asserts that it is hard to approximate the classical value of Unique-Label-Cover:
\begin{conjecture*}[UGC~\cite{khot_original}, informal]
    For all $\zeta,\delta>0$, the decision problem $\ULC_c(1-\zeta,\delta)$ is $\NP$-hard.
\end{conjecture*}

Kempe et al.~\cite{kempe} provided an approximation algorithm for the noncommutative value of Unique-Label-Cover (see Theorem~\ref{thm:kempe} for details), but this result does not address the quantum value. Thus, we propose the following quantum analogue of the Unique Games Conjecture (see Conjecture \ref{conj:qugc} for the precise statement):
\begin{conjecture*}[qUGC, informal]
    For all $\zeta,\delta>0$, the decision problem $\ULC_q(1-\zeta,\delta)$ is $\RE$-hard.
\end{conjecture*}
Assuming both UGC and qUGC, along with the algorithm from Kempe et al., the complexity landscape of Unique-Label-Cover can be summarized as shown in Figure \ref{fig:unique-label-cover-final-complexity-landscape}.
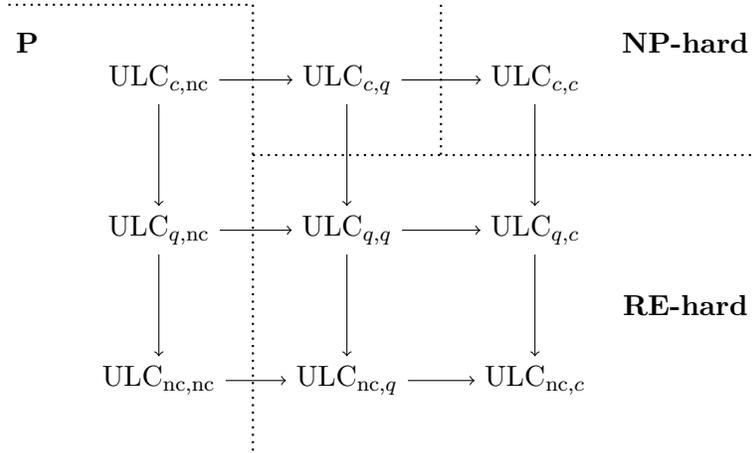
\begin{figure}
\centering
	\begin{tikzpicture}

        \node[] (c_nc) at (0,4) {$\ULC_{c,\nc}$};   \node[] (c_q) at (2.5,4) {$\ULC_{c,q}$};    \node[] (c_c) at (5,4) {$\ULC_{c,c}$};
        \node[] (q_nc) at (0,2) {$\ULC_{q,\nc}$};   \node[] (q_q) at (2.5,2) {$\ULC_{q,q}$};    \node[] (q_c) at (5,2) {$\ULC_{q,c}$};
        \node[] (nc_nc) at (0,0) {$\ULC_{\nc,\nc}$};\node[] (nc_q) at (2.5,0) {$\ULC_{\nc,q}$}; \node[] (nc_c) at (5,0) {$\ULC_{\nc,c}$};
        
        \draw[ ->] (c_nc) to (c_q);\draw[ ->] (c_q) to (c_c);
        \draw[ ->] (q_nc) to (q_q);\draw[ ->] (q_q) to (q_c);
        \draw[ ->] (nc_nc) to (nc_q);\draw[ ->] (nc_q) to (nc_c);

        \draw[ ->] (c_nc) to (q_nc);\draw[ ->] (c_q) to (q_q);\draw[ ->] (c_c) to (q_c);
        \draw[ ->] (q_nc) to (nc_nc);\draw[ ->] (q_q) to (nc_q);\draw[ -> ] (q_c) to (nc_c);

        \draw[dotted,thick] (-2,5) -- (1.25,5);
        \draw[dotted,thick] (1.25,5) -- (1.25,-1);
        \draw[dotted,thick] (1.25,3) -- (8,3);
        \draw[dotted,thick] (3.75,5) -- (3.75,3);
        \node[] at (7,4.5) {\bf NP-hard};
        \node[] at (7,1) {\bf RE-hard};
        \node[] at (-1.75,4.5) {\bf P};
    \end{tikzpicture}
    \caption{The complexity landscape of Unique-Label-Cover assuming both UGC and qUGC.}
\label{fig:unique-label-cover-final-complexity-landscape}
\end{figure} 

Assuming that any one of $\ULC_{q,q},\ULC_{q,c},\ULC_{\nc,q},\ULC_{\nc,c}$ is $\RE$-hard constitutes a variant of the qUGC. The strongest among these variants is the conjecture concerning $\ULC_{q,q}$. As discussed in Section~\ref{sec:label-cover}, there are many (indeed, infinitely many) variants of qUGC. When we refer to qUGC without additional qualifiers, we are specifically referring to the conjecture stated above. 

Given the reductions shown in Figure \ref{fig:unique-label-cover-final-complexity-landscape}, it is likely easier to prove hardness for $\ULC_{q,c}$ or $\ULC_{\nc,c}$ than to prove either UGC or qUGC. Conversely, devising a polynomial-time approximation algorithm for $\ULC_{c,q}$ seems more achievable than disproving either UGC or qUGC. Note that $\ULC_{c,q}$ is in $\NP$ (since $\ULC_{c,c}$ is in $\NP$).

\subsubsection{Applications to Hardness of Approximation}
The result by O’Donnell and Yuen~\cite{odonnell-yuen} on the quantum value of $3$XOR is stated in the next theorem.
\begin{theorem*}
For every $\delta > 0$, the decision problem $\XOR_{q}(1-\delta,\frac{1}{2}+\delta)$ is $\RE$-hard.
\end{theorem*}

This result was originally proven via a reduction from $\LC_{\nc}$. However, a reduction from $\LC_q$ suffices with some straightforward modifications. The hardness of $\LC_q$ follows directly from the $\MIP^*=\RE$ theorem (as shown in Figure~\ref{fig:label-cover-complexity-landscape}). Since a direct proof of $\LC_q$ is likely easier than the full proof of $\MIP^* = \RE$~\cite{ji_mip_re}, we propose Problem 5 in Section \ref{sec:future}.

This theorem mirrors H\r{a}stad’s celebrated result on the classical value of $3$XOR~\cite{hastad2}. Our current understanding of approximability of $3$XOR is summarized in Figure~\ref{fig:transition-3xor}. 

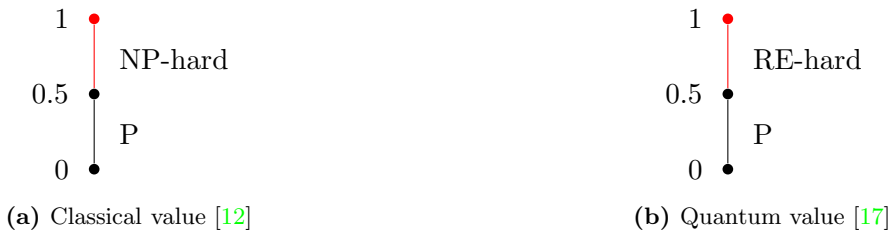
\begin{figure}[ht]
\centering
\begin{subfigure}{0.49\textwidth}
    \centering
    \begin{tikzpicture}
        \node[circle,fill=black, inner sep=0.05cm] (1) at (0,0) {};
        \node[circle,fill=black, inner sep=0.05cm] (2) at (0,1) {};
        \node[circle,fill=red, inner sep=0.05cm] (3) at  (0,2)  {};
        \draw[] (1) -- (2);
        \draw[red] (2) -- (3);
        
        \node at (1) [left=0.2cm]{$0$};
        \node at (2) [left=0.2cm]{$0.5$};
        \node at (3) [left=0.2cm]{$1$};

        \node at (1) [above right =0.2cm]{$\Pp$};
        \node at (2) [above right =0.2cm]{$\NP$-hard};        
    \end{tikzpicture}
    \caption{Classical value~\cite{hastad2}}
\end{subfigure}
\hfill
\begin{subfigure}{0.49\textwidth}
    \centering
    \begin{tikzpicture}
        \node[circle,fill=black, inner sep=0.05cm] (1) at (0,0) {};
        \node[circle,fill=black, inner sep=0.05cm] (2) at (0,1) {};
        \node[circle,fill=red, inner sep=0.05cm] (3) at  (0,2)  {};
        \draw[] (1) -- (2);
        \draw[red] (2) -- (3);
        
        \node at (1) [left=0.2cm]{$0$};
        \node at (2) [left=0.2cm]{$0.5$};
        \node at (3) [left=0.2cm]{$1$};

        \node at (1) [above right =0.2cm]{$\Pp$};
        \node at (2) [above right =0.2cm]{$\RE$-hard};        
    \end{tikzpicture}
    
    \caption{Quantum value~\cite{odonnell-yuen}}
\end{subfigure}
\caption{Approximability of classical and quantum values of $3$XOR: the vertical axis denotes approximation ratios between $0$ and $1$, with points of sharp transition in hardness ($0.5$ in both cases) indicated.}
\label{fig:transition-3xor}
\end{figure}

In Sections \ref{sec:2lin} and \ref{sec:maxcut}, we prove similar results for $2$-Lin and MaxCut. Our reductions are from $\ULC_q$ rather than $\LC_q$. An informal overview of our MaxCut result is provided in the next section.

\subsection{Result on MaxCut}\label{sec:result-on-quantum-maxcut}
How well can the classical, quantum, and noncommutative values of MaxCut be approximated? The current knowledge, including our result, is summarized in Figure \ref{fig:transition-max-cut}. 
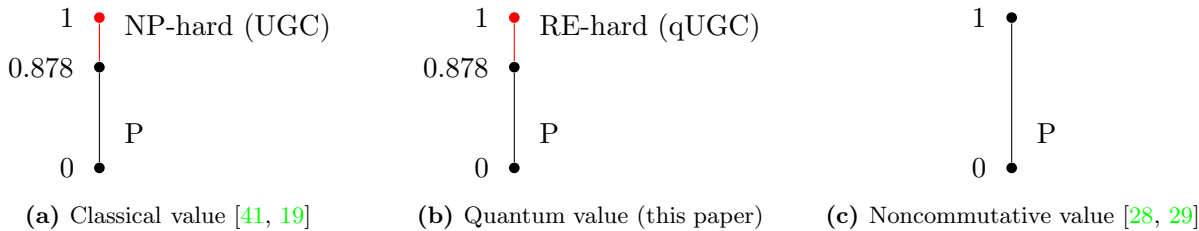
\begin{figure}[H]
\centering
\begin{subfigure}{0.32\textwidth}
    \centering
	\begin{tikzpicture}
        \node[circle,fill=black, inner sep=0.05cm] (1) at (0,0) {};
        \node[circle,fill=black, inner sep=0.05cm] (2) at (0,2*0.67) {};
        \node[circle,fill=red, inner sep=0.05cm] (3) at  (0,2)  {};
        \draw[black] (1) -- (2);
        \draw[red] (2) -- (3);
        
        \node at (1) [left=0.2cm]{$0$};
        \node at (2) [left=0.2cm]{$0.878$};
        \node at (3) [left=0.2cm]{$1$};

        \node at (1) [above right =0.2cm]{$\Pp$};
        \node at (2) [above right =0.2cm]{$\NP$-hard (UGC)};        
    \end{tikzpicture}
    \caption{Classical value~\cite{goemansmaxcut,kkmo}}
\end{subfigure}
\hfill
\begin{subfigure}{0.32\textwidth}
    \centering
	\begin{tikzpicture}
        \node[circle,fill=black, inner sep=0.05cm] (1) at (0,0) {};
        \node[circle,fill=black, inner sep=0.05cm] (2) at (0,2*0.67) {};
        \node[circle,fill=red, inner sep=0.05cm] (3) at  (0,2)  {};
        \draw[black] (1) -- (2);
        \draw[red] (2) -- (3);
        
        \node at (1) [left=0.2cm]{$0$};
        \node at (2) [left=0.2cm]{$0.878$};
        \node at (3) [left=0.2cm]{$1$};

        \node at (1) [above right =0.2cm]{$\Pp$};
        \node at (2) [above right =0.2cm]{$\RE$-hard (qUGC)};        
    \end{tikzpicture}
    \caption{Quantum value (this paper)}
\end{subfigure}
\hfill
\begin{subfigure}{0.32\textwidth}
    \centering
	\begin{tikzpicture}
        \node[circle,fill=black, inner sep=0.05cm] (1) at (0,0) {};
        \node[circle,fill=black, inner sep=0.05cm] (2) at  (0,2)  {};
        \draw[black] (1) -- (2);
        
        \node at (1) [left=0.2cm]{$0$};
        \node at (2) [left=0.2cm]{$1$};

        \node at (1) [above right =0.2cm]{$\Pp$};
    \end{tikzpicture}
    
    \caption{Noncommutative value~\cite{tsirelson1,tsirelson2}}
\end{subfigure}
\caption{Approximability of MaxCut across different types of values.}
\label{fig:transition-max-cut}
\end{figure}

The well-known Goemans-Williamson algorithm~\cite{goemansmaxcut} for the classical value of MaxCut achieves an approximation ratio of $\alpha_{\mathrm{gw}} \approx 0.878$. More specifically, they showed that $\omega_c(G) \geq \alpha_{\mathrm{gw}} \omega_{\mathrm{sdp}}(G)$, where $\omega_{\mathrm{sdp}}(G)$ is the optimal value of a certain semidefinite programming (SDP) relaxation for MaxCut. Meanwhile, Tsirelson~\cite{tsirelson1,tsirelson2} showed that $\omega_\nc(G) = \omega_{\mathrm{sdp}}(G)$. Combining these results yields the following chain of inequalities:
 \begin{equation}\label{eq:chain-of-inequalities}
\alpha_{\mathrm{gw}}\omega_{\mathrm{sdp}}(G) \leq \omega_c(G)  \leq  \omega_q(G) \leq \omega_\nc(G) = \omega_{\mathrm{sdp}}(G).
 \end{equation}
 \begin{figure}[H]
\centering
	\begin{tikzpicture}
        \node[circle,fill=black, inner sep=0.05cm] (1) at (-0.5,0) {};
        \node[circle,fill=black, inner sep=0.05cm] (2) at (1.2,0) {};
        \node[circle,fill=black, inner sep=0.05cm] (3) at (3,0) {};
        \node[circle,fill=black, inner sep=0.05cm] (4) at (4.6,0) {};
        \node[circle,fill=black, inner sep=0.05cm] (5) at  (8,0)  {};
        \draw[thick] (1) -- (2);
        \draw[thick] (2) -- (3);
        \draw[thick] (3) -- (4);
        \draw[thick] (4) -- (5);
        
        \node at (1) [below=0.2cm]{$0$};
        \node at (2) [below=0.2cm]{$\alpha_{\mathrm{gw}}\omega_{\mathrm{sdp}}(G)$};
        \node at (3) [above=0.2cm]{$\omega_c(G)$};
        \node at (4) [above=0.2cm]{$\omega_q(G)$};
        \node at (5) [above=0.2cm]{$\omega_\nc(G)$};
        \node at (5) [below=0.2cm]{$\omega_{\mathrm{sdp}}(G)$};
        
    \end{tikzpicture}

    \caption{The MaxCut interval}
\label{fig:values-on-the-interval}
\end{figure}
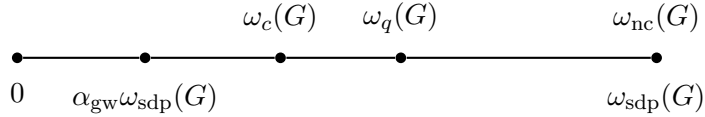 
From this chain of inequalities (see also Figure \ref{fig:values-on-the-interval}), we infer that the Goemans-Williamson algorithm is also a $0.878$-approximation algorithm for the quantum value of MaxCut. But can we do better for the quantum value? 

Under the assumption of the quantum Unique Games Conjecture introduced in this paper, we show:

\begin{center}
{\it The classical Goemans-Williamson algorithm is not only the best efficient approximation algorithm for the quantum value; it is the best approximation algorithm for the quantum value--regardless of efficiency. }
\end{center}

\begin{theorem*}[Theorem \ref{thm:maxcuthardness}, informal] Assuming the quantum Unique Games Conjecture, it is $\RE$-hard to approximate the quantum value of MaxCut to any ratio better than $\alpha_{\mathrm{gw}}$.
\end{theorem*}
Compare this with the following optimal inapproximability result from Khot et al.~\cite{kkmo}:
\begin{theorem*}[KKMO, informal] Assuming the Unique Games Conjecture, it is $\NP$-hard to approximate the classical value of MaxCut to any ratio better than $\alpha_{\mathrm{gw}}$.
\end{theorem*}

\textbf{Integrality Gap.} Similar to the Goemans-Williamson SDP relaxation, the quantum value--though undecidable--serves as a relaxation to the classical value. A natural question arises: how close is this quantum relaxation to the classical value? Can the quantum value outperform the SDP value in approximating the classical value?

Under the assumption of the quantum Unique Games Conjecture, we show that the SDP and quantum relaxations are incomparable. This is illustrated in Figure \ref{fig:extremes} and formally established in Theorem \ref{thm:integrality-gap}. Constructing examples that exhibit each of the extremes shown in Figure~\ref{fig:extremes} remains an open problem. 
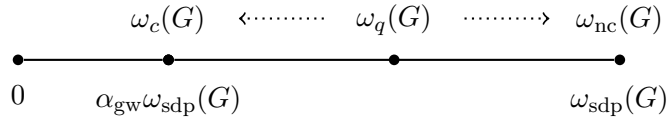
\begin{figure}[H]
\centering
	\begin{tikzpicture}
        \node[circle,fill=black, inner sep=0.05cm] (1) at (0,0) {};
        \node[circle,fill=black, inner sep=0.05cm] (2) at (2,0) {};
        \node[circle,fill=black, inner sep=0.05cm] (3) at (2,0) {};
        \node[circle,fill=black, inner sep=0.05cm] (4) at (2,0) {};
        \node[circle,fill=black, inner sep=0.05cm] (6) at (5,0) {};
        \node[circle,fill=black, inner sep=0.05cm] (5) at  (8,0)  {};
        \draw[thick] (1) -- (2);
        \draw[thick] (2) -- (3);
        \draw[thick] (3) -- (4);
        \draw[thick] (4) -- (5);

        \node[] (3above) at (2,0.8) {};
        \node[] (6above) at (5,0.8) {};
        \node[] (5above) at (8,0.8) {};

        \node[] (aux3above) at (2.8,0.55) {};
        \node[] (aux6above) at (4.2,0.55) {};
        \node[] (aux6above2) at (5.8,0.55) {};
        \node[] (aux5above) at (7.2,0.55) {};

        \node at (1) [below=0.2cm]{$0$};
        \node at (2) [below=0.2cm]{$\alpha_{\mathrm{gw}}\omega_{\mathrm{sdp}}(G)$};
        \node at (3) [above=0.2cm]{$\omega_c(G)$};
        \node at (6) [above=0.2cm]{$\omega_q(G)$};
        \node at (5) [above=0.2cm]{$\omega_\nc(G)$};
        \node at (5) [below=0.2cm]{$\omega_{\mathrm{sdp}}(G)$};   

        \draw[thick,dotted, ->] (aux6above) to (aux3above); 
        \draw[thick,dotted, ->] (aux6above2) to (aux5above); 
    \end{tikzpicture}
    \caption{There exist instances realizing both extremes depicted in this picture. That is, among the integrality gap instances of MaxCut (see Feige and Schechtman~\cite{feige_integrality}), i.e.\ those instances with $\omega_c(G)=\alpha_{\mathrm{gw}}\omega_{\mathrm{sdp}}(G)$, there are instances for which $\omega_q(G) = \omega_c(G)$ and instances for which $\omega_q(G) = \omega_{\mathrm{sdp}}(G)$.}
\label{fig:extremes}
\end{figure}

\subsection{Proof Ideas}\label{sec:proof-ideas}
In the conventional $\NP$ proof system for Label-Cover, the certificate is simply the list of labels for the vertices. This is what we called \emph{classical assignment} in Section \ref{sec:intro-quantum-value}. The verifier’s task in this proof system is to check that the assigned labels satisfy the edge constraints.

Håstad’s gap-preserving reduction from Label-Cover (used, for instance, in proving the hardness of $3$SAT and $3$XOR) is best viewed as the construction of a more efficient version of this $\NP$ proof system, known as a probabilistically checkable proof (PCP). The PCP certificate for Label-Cover is an encoding of the classical assignment in binary using an error-correcting code, commonly referred to as the Long-Code, pioneered in~\cite{bellare}. In addition to verifying that the labels in the encoding satisfy the edge constraints, the PCP verifier must also check that the proof consists of valid codewords. The soundness of this PCP relies on Fourier analysis. This was one of the key insights of Håstad's celebrated work.

In our case, we are reducing Unique-Label-Cover to $2$-Lin (Section \ref{sec:2lin}) and MaxCut (Section \ref{sec:maxcut}), but the core principles remain the same. Our reductions are based on the classical reductions by Khot~\cite{khot_original} and KKMO~\cite{kkmo}, respectively. There are, however, several key differences, which we discuss in the rest of this section. For concreteness, we focus on the reduction from Unique-Label-Cover to MaxCut.

\textbf{Completeness.} This step requires a generalization of the Long-Code to encode quantum assignments. This is a straightforward generalization of the Long-Code to the operator setting and it is due to O'Donnell and Yuen \cite{odonnell-yuen}. The main challenge in proving completeness is ensuring that the encoded quantum assignment satisfies the commutation requirements of a quantum MaxCut assignment.

\textbf{Soundness.} The proof of soundness, as is customary, is done contrapositively. We begin with a quantum assignment for the MaxCut instance that achieves a large value, and from this, we construct a quantum assignment that has a sufficiently large value in the Unique-Label-Cover instance.

The classical proof constructs a probabilistic labeling based on the Fourier coefficients of the MaxCut assignment. Given a MaxCut assignment $f$, the labeling samples a subset of labels $S$ with probability proportional to the squared Fourier coefficients, $\hat{f}(S)^2$. It then probabilistically selects a label $a \in S$ for each vertex. In the quantum setting, these squared Fourier weights are positive semidefinite operators. Defining $P^a \coloneqq \sum_{S:a \in S}\frac{1}{|S|} \hat{f}(S)^2$ for each label $a$ results in a \emph{POVM measurement} for each vertex in the Unique-Label-Cover instance (see Section \ref{sec:measurements} for the definition of POVMs). This is also due to~\cite{odonnell-yuen}.

These measurements satisfy the commutation relations required for a quantum assignment. However, they are not projective measurements, as required by definition of quantum assignment. Naimark’s dilation theorem (see, for example, Watrous~\cite{watrous-book}) is typically used in such cases to convert POVMs into projective measurements. However, due to the commutation requirements (and the use of tracial states), Naimark's theorem is not applicable. Therefore, we need to projectivize these POVMs while preserving the commutation relations, which we accomplish using the projectivization lemma proved in Section \ref{sec:projectivization}.

Classically, the proof of soundness relies on a deep result from Boolean Fourier analysis, known as the Majority-Is-Stablest~\cite{MOO}, which relates Boolean functions to their Fourier weights. If the operators in our assignments were simultaneously diagonalizable, we could apply this theorem directly, as in the classical proof. However, the inherent noncommutativity of quantum assignments prevents us from doing so, forcing us to take a new approach. 

To explain the problem and our resolution at a high level, recall that our goal is to construct a quantum labeling for an instance of Unique-Label-Cover. Let $u$ be a vertex from the left-vertex set, and let $N_u$ denote its neighbors in the right-vertex set. In the operator setting, while each neighbor's operator individually commutes with $u$'s, the operators assigned to the neighbors may not commute with each other. This noncommutativity is a fundamental obstacle.

We first encountered this issue earlier, when attempting to view quantum assignments as extensions of probabilistic assignments.  In a classical probabilistic assignment, each vertex $u$ is assigned a single probability measure $P_u$ over the set of labels. In contrast, a quantum assignment associates two probability measures, $P_{u,e}$ and $P_{v,e}$, with each edge $e = (u,v)$. The crux of the issue is that different edges incident on $u$—say, $e$ and $e'$—can impose vastly different probability measures on $u$, which may conflict with each other. 

This same phenomenon arises in the proof of soundness, except it manifests not in the probability measures, but in the Fourier coefficients of the quantum assignment. Classical soundness proofs organize around sets of neighbors, but this approach is too coarse to be applicable to the operator setting. The key to resolving this issue is to structure the proof around individual edges; this allows us to manage the noncommutative nature of the operator assignments without the sort of conflicts we mentioned earlier.

\subsection{Future Directions}\label{sec:future}
There are many open problems concerning the quantum value of CSPs and the conjectures we have proposed regarding Unique-Label-Cover:
\begin{enumerate}
    \item \textbf{Quantizing Classical Reductions: } The works of O'Donnell and Yuen~\cite{odonnell-yuen}, as well as this paper, establish straightforward translations of some of the most well-known classical inapproximability results to the quantum setting. Could this indicate a more general phenomenon? More precisely, can any classical gap-preserving reduction from Label-Cover (or its variants like Unique-Label-Cover or Smooth-Label-Cover) be extended to a reduction for the quantum value?
    
    \item \textbf{UGC versus qUGC: } What is the relationship between the classical UGC and the quantum UGC? If UGC is proven via a reduction from Label-Cover, can this reduction be ``quantized" to prove qUGC? A positive answer to the previous question would also resolve this one affirmatively.
    
    \item \textbf{Variants of qUGC: } In the classical setting, several formulations of UGC have been shown to be equivalent (see Section 3 in~\cite{khot_survey2} and references therein). Could similar equivalences hold for the quantum variants discussed in this paper? See Section \ref{sec:label-cover} for further discussion of these qUGC variants.

    \item \textbf{Noncommutative CSPs: } This work focuses primarily on the quantum value of CSPs. Could variants of qUGC help resolve hardness of approximation questions for the noncommutative value? The noncommutative value, which also plays an important role in quantum information, exhibits some intriguing behavior. As shown in Figure \ref{fig:transition-max-cut}, unlike the classical or quantum value, the noncommutative value of MaxCut can be computed in polynomial time. However, it is known that the noncommutative value of Max-$3$-Cut is undecidable \cite{ji2013binary,ji_mip_re}. \cite{original} provides a $0.864$-approximation algorithm for the noncommutative value of Max-$3$-Cut. We suspect that a variant of qUGC could prove useful in establishing the optimality of this algorithm. 

    \item \textbf{Weaker Version of $\mathbf{MIP^*=RE}$:} Is there a direct way to prove the hardness of approximating the quantum value of Label-Cover without relying on $\MIP^* = \RE$? One potential approach is to quantize Dinur’s proof of the classical PCP theorem~\cite{arora1,arora2,raz,DinurPCP}. The commutation relations in $\LC_q$ could make this quantization feasible. Could this weaker result also have some of the implications of $\MIP^* = \RE$ such as the resolution of the Connes Embedding Problem?
    
    \item \textbf{Integrality Gap Instances for MaxCut:} What is the smallest ratio between the classical and quantum values across all MaxCut instances? Assuming quantum UGC, we show this ratio is the Goemans-Williamson constant $\alpha_{\mathrm{gw}} \approx 0.878$ (Theorem \ref{thm:integrality-gap}). Can we construct an instance attaining this ratio?

    \item \textbf{Label-Cover for QMA: } Approximating the classical value of Label-Cover is $\NP$-complete (PCP theorem). Similarly, approximating the noncommutative and quantum values of Label-Cover is $\RE$-complete ($\MIP^* = \RE$ theorem). What version of Label-Cover or Unique-Label-Cover naturally captures $\QMA$? Could further constraints on the variations of Label-Cover and Unique-Label-Cover explored in this paper lead to suitable candidates? A positive resolution to this problem would answer the Quantum Games PCP Conjecture~\cite{anand-games-pcp}.
        
    \item \textbf{Hardness of Approximation in Hamiltonian Complexity: } In quantum information theory, quantum MaxCut often refers to the anti-ferromagnetic Heisenberg model, a special case of the local Hamiltonian problem. Although~\cite{quantum-max-cut-integrality-gap} attempts to prove UGC hardness for the anti-ferromagnetic Heisenberg model, reductions from UGC only establish $\NP$-hardness. Could variations of qUGC yield stronger results? 
    
    While the inapproximability of the anti-ferromagnetic Heisenberg model is one of our motivations for proposing the quantum UGC, our current methods do not yet apply to the inapproximability of local Hamiltonian problems.
    
    \item \textbf{Gadget Reductions: } In classical complexity theory, gadgets~\cite{TrevisanGadgets,hastad2} are widely used to demonstrate hardness of approximation. Can this method be ``quantized"?~\cite{ji2013binary,harris2023universality} are examples of gadget reductions in the quantum setting.
\end{enumerate}

\section*{Organization of the Paper} 
In Section \ref{sec:preliminaries}, we present notations, definitions, and 
some Fourier analytic results that underpin this work. Section \ref{sec:notations} specifically addresses the notations used throughout the paper. In Section \ref{sec:measurements}, we introduce the concept of quantum measurement and discuss several definitions related to projective measurements. Section \ref{sec:computational-problems} reviews concepts such as reductions and decision problems, and Section \ref{sec:mis-bourgain} summarizes results from Fourier analysis.

Section \ref{sec:general-csps} offers the most general definition of CSPs necessary for this work. It also covers the classical, quantum, and noncommutative values of a CSP instance, providing definitions which are crucial for the remainder of the paper. In Section \ref{sec:label-cover}, we discuss variants of quantum UGC. Section \ref{sec:2lin} contains the proof of our theorem on the quantum value of $2$-Lin, and Section \ref{sec:maxcut} presents the proof of our theorem on the quantum value of MaxCut. Finally, two technical lemmas are given in Sections \ref{sec:projectivization} and \ref{sec:folding-lemma}.

\textbf{Acknowledgement.} We thank Henry Yuen for valuable discussions on the topics explored in this paper. We thank Eric Culf for the observation leading to the proof of Theorem \ref{thm:integrality-gap}.

\section{Preliminaries}\label{sec:preliminaries}
\subsection{Notations}\label{sec:notations}
For every $n \in \N$ we let $[n] = \{1,\ldots,n\}$. When $x$ is a function from $[n]$, we may treat it as a vector and write $x_i \coloneqq x(i)$. For $\pi\in S_n$, a permutation, we let $x\circ \pi$ be the vector for which $(x\circ \pi)_i = x_{\pi(i)}$, for all $i$. 

Hilbert spaces in this work are always assumed to be finite dimensional vector spaces $\C^d$ for some $d\in \N$. We denote the set of operators on the Hilbert space $\mathcal{H}$ by $\mathcal{M}(\mathcal{H})$. The unitary group on $\mathcal{H}$ is denoted by $\U(\mathcal{H})$. Given a Hilbert space, we use $\tr$ to denote its dimension-normalized trace. Given two operators $A$ and $B$, their commutator is denoted by $[A,B]\coloneqq AB - BA$.

\subsection{Measurements}\label{sec:measurements}
For $m \in \N$, an $m$-outcome POVM $P$ acting on a Hilbert space $\mathcal{H}$ is a set of positive semidefinite (PSD) operators $\{P^1,\ldots,P^m\}$ in $\M(\mathcal{H})$ that sum to the identity operator. POVM stands for \emph{positive operator-valued measure}, which captures the most general notion of measurement in quantum mechanics. We need two definitions concerning POVMs.
\begin{definition}[Simultaneous measurability]\label{def:simultaneous-measurability}
A set of POVMs $\{P_i^a\}_{a\in [m]}$, for $i \in [n]$, is said to be simultaneously measurable, or commuting for short, if $[P_i^a,P_j^b] = P_i^a P_j^b - P_j^b P_i^a = 0$ for all $a,b\in [m]$ and $i\neq j$. 
\end{definition}
\begin{definition}[Self-commuting measurements]\label{def:self-commuting}
A POVM is said to be self-commuting if all of its operators commute.
\end{definition}
Throughout this paper, all POVMs are self-commuting. We denote the set of $m$-outcome self-commuting POVMs on $\mathcal{H}$ by $\POVM_m(\mathcal{H})$. Operators of simultaneously measurable self-commuting POVMs can all be diagonalized in the same basis. 

PVMs, short for projective-valued measures, are a special case of POVMs, and indeed a special case of self-commuting POVMs. A POVM $\Pi = \{\Pi^1,\ldots,\Pi^m\}$ is said to be a PVM if each $\Pi^a$ is a projection operator, i.e.\ $(\Pi^a)^* = \Pi^a$ and $(\Pi^a)^2 = \Pi^a$. We let $\PVM_m(\mathcal{H})$ denote the set of $m$-outcome PVMs on $\mathcal{H}$. Almost all POVMs appearing in this paper are PVMs with the exception of a few places where self-commuting POVMs appear.  

A unitary $X$ that is also Hermitian is called an \emph{observable}. The set of all observables is denoted by $\Obs(\mathcal{H})$. Observables and binary-outcome PVMs are in one-to-one correspondence: for every observable $X$, operators $\Pi^{+1} \coloneqq (I+X)/2$ and $\Pi^{-1} \coloneqq (I-X)/2$ form a PVM. A set of binary-outcome PVMs are simultaneously measurable if and only if their corresponding observables commute. 

\textbf{Indexing measurement operators.} As suggested, the binary set $\{-1,+1\}$ (the set of eigenvalues of observables) is used to index operators in a binary-outcome PVM. When $m \geq 3$, the set $[m]$ is used to index operators in an $m$-outcome PVM. This choice helps to make our equations easier to parse. See also Section \ref{sec:definitions-lin-maxcut}.

\subsection{Computational Problems}\label{sec:computational-problems}
For us, a \emph{computational problem} is a pair $(\P,\omega)$ where $\P$ is the set of instances and $\omega:\P \to [0,1]$ is the function to be computed. We often write $\P_\omega$ for $(\P,\omega)$. In this paper we call $\omega$ the \emph{value} (function) and when it is clear from the context we write $\P$ to refer to the computational problem itself.  

We now define the \emph{decision version} of a computational problem.
\begin{definition}\label{def:decision-problem}
For every $0 < \delta < \gamma \leq 1$, the decision problem of $(\P,\omega)$, denoted by $\P_\omega(\gamma,\delta)$, is the problem: given an instance $\phi \in \P$ with the promise that either
\begin{itemize}
    \item $\omega(\phi) \geq \gamma$, called a ``yes'' instance, or
    \item $\omega(\phi) < \delta$, called a ``no'' instance,
\end{itemize} 
decide whether the instance is a yes or no instance. We refer to $\gamma$ as the completeness parameter and $\delta$ as the soundness parameter. 

If distinct values $\omega_1$ and $\omega_2$ are used for yes and no instances, that is if
\begin{itemize}
    \item $\omega_1(\phi) \geq \gamma$, is a yes instance, and
    \item $\omega_2(\phi) < \delta$, is a no instance,
\end{itemize} 
then the corresponding decision problem is denoted $\P_{\omega_1,\omega_2}(\gamma,\delta)$. Thus $\P_{\omega}(\gamma,\delta)$ is $\P_{\omega,\omega}(\gamma,\delta)$.
\end{definition}

\begin{definition}\label{def:reduction}
A reduction from $\P_{\omega_1,\omega_2}(\gamma,\delta)$ to $\P'_{\omega'_1,\omega'_2}(\gamma',\delta')$ is a deterministic polynomial time map $r$ that sends every instance in $\P$ to an instance in $\P'$. Furthermore $r$ needs to be complete and sound. Completeness is the property that a yes instance is sent to a yes instance. The soundness is the property that a no instance is sent to a no instance. When a reduction exists we may write $\P_{\omega_1,\omega_2}(\gamma,\delta) \longrightarrow \P'_{\omega'_1,\omega'_2}(\gamma',\delta')$.
\end{definition}

\subsection{Fourier Analysis}\label{sec:mis-bourgain}
We quickly review the basic definitions and notations from Fourier analysis on hypercube. For a proper introduction, refer to the textbook \cite{odonnell_book}. Let $g: \{\pm 1\}^m \rightarrow \R$ be a function and let $\{\pm 1\}^m$ be equipped with the uniform probability measure. For every $S\subseteq [m]$, the \emph{characteristic function} $\chi_S: \{\pm 1\}^m \rightarrow \{\pm 1\}$ is given by $\chi_S(x)=\prod_{a \in S} x_a$. This set forms an orthonormal basis, that is for every $S,T \subseteq [m]$
\begin{align*}
    \expect_x \Brac{ \chi_S(x) \chi_T(x)} = \delta_{S,T},
\end{align*}
where $\delta_{S,T}$ is the Kronecker delta. For all $S\subseteq [m]$ and $x\in \{\pm 1\}^m$, the Fourier coefficient of $g$ at $S$ is
\begin{align*}
    \hat{g}(S) = \expect_x \Brac{\chi_S(x) g(x)}.
\end{align*}
The Fourier inversion formula is
\begin{align*}
    g(x)=\sum_{S\subseteq [m]} \chi_S(x) \hat{g}(S).
\end{align*}

\begin{definition}\label{def:influence}
    For $a\in [m]$, the influence of the $a$'th index on $g$, denoted $\Infl_a(g)$, is 
\begin{align*}
    \Infl_a(g) = \expect_x \Brac{\mathrm{Var}_{x_a}\Brac{g(x)}}= \sum_{S: a \in S} \hat{g}(S)^2.
\end{align*}
The $k$-degree influence of the index $a$ on $g$ is
\begin{align*}
    \Infl_a^{\leq k}(g) = \sum_{S: a\in S, \abs{S}\leq k} \hat{g}(S)^2.
\end{align*}
\end{definition}

\begin{definition}\label{def:noisestab}
    Let $-1<\rho\leq 0$. The noise stability of $g$ at $\rho$ is
    $$S_\rho(g)=\expect_{x,\mu}\Brac{g(x) g(x\mu)},$$
    where $x$ is sampled from the uniform distribution, and $\mu$ is sampled independently such that $\mu_a=1$ with probability $\frac{1+\rho}{2}$ and $\mu_a=-1$ with probability $\frac{1-\rho}{2}$.
\end{definition}

We can now state two Fourier analytic theorems used in our proofs.
\begin{theorem}[Bourgain's Junta Theorem~\cite{Bourgain02}]\label{thm:Bourgain}
    Let $g: \{\pm 1\}^m \to \{\pm 1\}$ be any Boolean function and $k>0$ an integer. Then for every $\frac{1}{2}<t<1$, there exists a constant $c_t>0$ such that
    \begin{align*}
        \textit{if} \quad \sum_{S: \abs{S}>k} \hat{g}(S)^2 <c_t k^{-t} \quad \textit{then} \quad \sum_{S: \abs{\hat{g}(S)} \leq \frac{1}{10}4^{-k^2}} \hat{g}(S)^2 < \frac{1}{100}.
    \end{align*}
\end{theorem}

\begin{theorem}[Majority is Stablest (MIS)~\cite{MOO}]\label{thm:MIS}
    For every $\rho \in (-1,0]$ and $\varepsilon>0$, there is a small enough $\delta=\delta(\varepsilon,\rho)>0$ and a large enough $k=k(\varepsilon,\rho)$ such that if $g:\{\pm 1\}^m \rightarrow [-1,1]$ is any function with $\Infl_a^{\leq k}(g) \leq \delta,$
    for all $a\in [m]$, then 
    $$S_\rho(g) \geq 1 - \frac{2}{\pi} \arccos{\rho} - \varepsilon.$$
\end{theorem}

Even though the Fourier analysis of functions on the Boolean hypercube does not depend on the range, we frequently encounter functions $\alpha: \{\pm 1\}^m \to \Obs(\mathcal{H})$ so it is useful to address a few minor issues that arise. Let $\alpha: \{\pm 1\}^m \to \Obs(\mathcal{H})$ where $\mathcal{H}$ is a finite-dimensional Hilbert space. For all $S\subseteq [m]$ and $x\in \{\pm 1\}^m$, the Fourier transform and its inverse are defined exactly as before
\begin{align*}
    \hat{\alpha}(S) = \expect_x \Brac{\chi_S(x) \alpha(x)} \quad \text{and} \quad \alpha(x)=\sum_{S\subseteq [m]} \chi_S(x) \hat{\alpha}(S).
\end{align*}
Moreover, the Parseval identity also holds as in the scalar case
\[\sum_{S\subseteq [m]} \hat{\alpha}(S)^2 = I.\]
It is clear that since $\alpha(x)$ are observables, the operators $\hat{\alpha}(S)$ are Hermitian, and thus $\hat{\alpha}(S)^2$ are PSD. Consequently, the set of squared Fourier coefficients $\{\hat{\alpha}(S)^2\}_{S\subseteq [m]}$ defines a POVM.

We say $\alpha: \{\pm 1\}^m \to \Obs(\hilb)$ is odd if for every $x\in \{\pm 1\}^m$, $\alpha(x)=-\alpha(-x)$. A simple calculation shows that $\hat{\alpha}(\emptyset) = 0$ whenever $\alpha$ is odd.

\begin{lemma}\label{lem:povm_from_obs}
    Let $\alpha: \{\pm 1\}^m \to \Obs(\mathcal{H})$ be a function defined on a Hilbert space $\mathcal{H}$. Consider the operators 
    \begin{align*}
        P^a = \sum_{S: a \in S} \frac{1}{\abs{S}} \hat{\alpha}(S)^2,
    \end{align*}
    for all $a\in [m]$. Then the set $P=\{P^a\}_{a\in [m]} \cup \{Q\}$ where $Q= I -\sum_a P^a$ is a POVM. If in addition $\alpha$ is odd, then $Q=0$. If observables $\alpha(x)$ commute for all $x$, then $P$ is a self-commuting POVM. 
\end{lemma}

\begin{proof}
    By Parseval, we have
    \begin{align*}
        \sum_a P^a = \sum_a \sum_{S: a \in S} \frac{1}{\abs{S}} \hat{\alpha}(S)^2 = \sum_{S : S \neq \emptyset}\hat{\alpha}(S)^2 \leq I,
    \end{align*}
    Therefore $Q= I - \sum_a P^a$ is a PSD matrix, and we have a POVM. When $\alpha$ is odd, we have $\hat{\alpha}(\emptyset) = 0$ and therefore $Q=I -\sum_{S : S \neq \emptyset}\hat{\alpha}(S)^2 = 0$ by Parseval's identity. Since Fourier coefficients are linear combinations of $\alpha(x)$'s, the self-commutation of the POVM follows trivially. 
\end{proof}

\section{CSPs}\label{sec:csps}
\subsection{CSPs, Assignments, and Values}\label{sec:general-csps}
All computational problems in this paper are examples of constraint satisfaction problems (CSPs). We refer to $\CSP{k}(m)$ as the $k$-ary CSP over an alphabet of size $m$. We now define an instance of $\CSP{k}(m)$.
\begin{definition}\label{def:general-csp}
For $k,m\in \N$, an instance in $\CSP{k}(m)$ consists of
\begin{itemize}
\item a pair $(V,E)$, where $V$ is a set of variables and $E$ is a multiset of $k$-tuples of $V$ representing constraints, 
\item a set of \emph{predicates} $f_e:[m]^k \to \{0,1\}$ for every $e \in E$, and
\item a probability distribution $p$ on $E$. We refer to $p_e \coloneqq p(e)$ as the weight of the constraint $e$.
\end{itemize}
We denote this instance with the tuple $(V,E,f,p)$. The set $[m]$ is referred to as the alphabet or the set of labels. We let $p_e = 0$ whenever $e \notin E$. We always let $n \coloneqq |V|$.
\end{definition}
\textbf{Remark.} Note that $E$ is a multiset. This means that the same $k$-tuple of $V$, say $e$, may appear multiple times in $E$ each time with a different predicate. When $k=2$, the pair $(V,E)$ can be viewed as a directed graph (possibly with parallel edges since $E$ is a multiset).

The arity $k$ is always fixed, but in the context of Label-Cover, as we will soon see, it is customary to let the alphabet size $m$ vary. We denote by $\CSP{k}$ the union of $\CSP{k}(m)$ for all $m$. We assume a natural representation of instances is fixed. The size of the instance is the length of such a representation, and it is clearly a $\poly(n,m)$.

We next define the \emph{classical and quantum values} of instances in $\CSP{k}(m)$.
\begin{definition}[Classical assignment and classical value]\label{def:classical-value}
Given an instance $\phi$ in $\CSP{k}(m)$, a (classical) assignment is a function $a:V \to [m]$. 

The value of this assignment on the constraint $e=(i_1,\ldots,i_k)$ is $f_e(a_{i_1},\ldots,a_{i_k})$. 

The value of this assignment on the instance, denoted $\omega(\phi,a)$, is the weighted sum $$\sum_{e=(i_1,\ldots,i_k)\in E} p_e f_e(a_{i_1},\ldots,a_{i_k}).$$ 

Finally, the (classical) value of the instance, denoted $\omega_c(\phi)$, is the maximum of $\omega(\phi,a)$ over all classical assignments.
\end{definition}

\begin{definition}[Quantum assignment and quantum value]\label{def:quantum-value}
A quantum assignment consists of a Hilbert space $\mathcal{H}$ and a function $\Pi:V \to \PVM_m(\mathcal{H})$ such that for every constraint $e=(i_1,\ldots,i_k) \in E$ the PVMs $\Pi_{i_1},\ldots,\Pi_{i_k}$ are simultaneously measurable.  

The value of this assignment on the constraint $e=(i_1,\ldots,i_k)$ is $$\sum_{a \in [m]^k} f_e(a_1,\ldots,a_k)\tr\Paren{\Pi_{i_1}^{a_1}\cdots \Pi_{i_k}^{a_k}}.$$ 

The value of the assignment on the instance, denoted $\omega(\phi,\Pi)$, is the weighted sum $$\sum_{e=(i_1,\ldots,i_k)\in E} p_e\sum_{a \in [m]^k} f_e(a_1,\ldots,a_k)\tr\Paren{\Pi_{i_1}^{a_1}\cdots \Pi_{i_k}^{a_k}}.$$ 

Finally, the quantum value of the instance, denoted $\omega_q(\phi)$, is the supremum of $\omega(\phi,\Pi)$ over all quantum assignments. 
\end{definition}
Note that the quantum value is well-defined since the traces $\tr\left(\Pi_{i_1}^{a_1} \cdots \Pi_{i_k}^{a_k}\right)$ are real numbers for every constraint $e = (i_1, \ldots, i_k)$ due to the commutation relations imposed in the definition of the quantum assignment. When $\mathcal{H}$ is one-dimensional, a quantum assignment reduces to a classical assignment, implying the inequality $\omega_q(\phi) \geq \omega_c(\phi)$.

We now define the \emph{noncommutative value} of an instance in $\CSP{2}(m)$. The key difference between noncommutative and quantum assignments is that the simultaneous measurability constraint is no longer required.

\begin{definition}[Noncommutative assignment and noncommutative value]\label{def:noncommutative-value}
Given an instance $\phi$ in $\CSP{2}(m)$, a noncommutative assignment consists of a Hilbert space $\mathcal{H}$ and a function $\Pi:V \to \PVM_m(\mathcal{H})$. 

The value of the assignment on the constraint $e=(i,j)$ is $$\sum_{a,b \in [m]} f_e(a,b)\tr\Paren{\Pi_{i}^{a}\Pi_{j}^{b}}$$

The value of the assignment on the instance, denoted $\omega(\phi,\Pi)$, is the weighted sum $$\sum_{e=(i,j)\in E} p_e\sum_{a,b \in [m]} f_e(a,b)\tr\Paren{\Pi_{i}^{a}\Pi_{j}^{b}}.$$ 

Finally, the noncommutative value of the instance $\omega_\nc(\phi)$ is the supremum of $\omega(\phi,\Pi)$ over all noncommutative assignments. 
\end{definition}
\textbf{Remark.} While the classical and quantum values are defined for all instances, the noncommutative value is defined only for instances in $2$-CSPs. This is because, for $k \geq 3$, the traces $\tr\left(\Pi_{i_1}^{a_1} \cdots \Pi_{i_k}^{a_k}\right)$ are not necessarily real numbers when $\Pi_{i_1}, \ldots, \Pi_{i_k}$ are not simultaneously measurable.

Every classical assignment is a quantum assignment and every quantum assignment is a noncommutative assignment. Consequently, we have the chain of inequalities $\omega_\nc(\phi) \geq \omega_q(\phi) \geq \omega_c(\phi)$. 

\, 

We conclude this section by noting that infinitely many types of quantum assignments and values can be defined.

\noindent\textbf{A Spectrum of Quantum Values.} So far, we have introduced three types of assignments for constraint satisfaction problems, based on the degree of noncommutativity allowed between measurement operators. At one extreme, classical assignments allow only fully commuting measurements, while at the other, noncommutative assignments impose no commutation relations. Quantum assignments lie in between these two extremes. More generally, we can conceive of a spectrum of quantum assignments, parameterized by the \emph{degree of noncommutativity} allowed between measurement operators:

\begin{itemize}
    \item Imposing fewer commutation relations leads to what we call strong-quantum assignments, with the corresponding value, denoted $\omega_\sq$, shifting closer to the noncommutative value.\footnote{We do not formally define the strong-quantum value here.}
    \item Conversely, imposing more commutation (e.g., requiring neighbors of neighbors to commute in MaxCut) leads to weak-quantum assignments. The corresponding value, denoted $\omega_\wq$, shifts closer to the classical value. See Definition \ref{def:weak-quantum-value} for details. 
\end{itemize}
If for an instance $\phi$ all five types of values are defined, we have\[\omega_c(\phi) \leq \omega_\wq(\phi) \leq \omega_q(\phi) \leq \omega_\sq(\phi) \leq \omega_\nc(\phi).\]
 For $s,t \in \{c,\wq, q, \sq, \nc\}$, we let $\P_{s,t}$ be a shorthand for the decision problem $\P_{\omega_s,\omega_t}$ (see Section \ref{sec:computational-problems} for this notation). When $s = t$ we use the shorthand $\P_s$.

\textbf{Trivial Reductions.} Let $\P$ be a CSP for which all five types of values are defined. Next we explore an example of a reduction between decision versions of $\P$. We can show that, for every choice of $\zeta,\delta > 0$, the decision problem $\P_{q,\nc}(1-\zeta,\delta)$ reduces to $\P_{\nc,\nc}(1-\zeta,\delta)$. The reduction is the identity map on the set of  instances in $\P$: there is nothing to prove for soundness, and the completeness follows from the fact that $\omega_\nc(\phi) \geq \omega_q(\phi)$ for every $\phi\in\P$. We say $\P_{q,\nc}$ trivially reduces to $\P_{\nc,\nc}$ and denote this by $\P_{q,\nc} \longrightarrow \P_{\nc,\nc}$.

There are many more examples of such trivial reductions between decision variants of $\P$. Suppose an abstract ordering $c < \wq < q < \sq < \nc$ is fixed on the set of symbols $\{c,\wq,q,\sq,\nc\}$. Then, for all $s,t,s',t' \in \{c,\wq, q, \sq, \nc\}$ there is a trivial reduction $\P_{s,t} \longrightarrow \P_{s',t'}$ whenever $s \leq s'$ and $t \geq t'$. This is illustrated in Figure \ref{fig:expanded-trivial-reductions}.

\begin{figure}[ht]
\centering
	\begin{tikzpicture}

        \node[] (c_nc) at (0,8) {$\P_{c,\nc}$}; \node[] (c_sq) at (2,8) {$\P_{c,\sq}$};   \node[] (c_q) at (4,8) {$\P_{c,q}$};  \node[] (c_wq) at (6,8) {$\P_{c,\wq}$};  \node[] (c_c) at (8,8) {$\P_{c,c}$};
        \node[] (wq_nc) at (0,6) {$\P_{\wq,\nc}$}; \node[] (wq_sq) at (2,6) {$\P_{\wq,\sq}$};   \node[] (wq_q) at (4,6) {$\P_{\wq,q}$};  \node[] (wq_wq) at (6,6) {$\P_{\wq,\wq}$};  \node[] (wq_c) at (8,6) {$\P_{\wq,c}$};
        \node[] (q_nc) at (0,4) {$\P_{q,\nc}$}; \node[] (q_sq) at (2,4) {$\P_{q,\sq}$};   \node[] (q_q) at (4,4) {$\P_{q,q}$};  \node[] (q_wq) at (6,4) {$\P_{q,\wq}$};  \node[] (q_c) at (8,4) {$\P_{q,c}$};
        \node[] (sq_nc) at (0,2) {$\P_{\sq,\nc}$}; \node[] (sq_sq) at (2,2) {$\P_{\sq,\sq}$};   \node[] (sq_q) at (4,2) {$\P_{\sq,q}$};  \node[] (sq_wq) at (6,2) {$\P_{\sq,\wq}$};  \node[] (sq_c) at (8,2) {$\P_{\sq,c}$};
        \node[] (nc_nc) at (0,0) {$\P_{\nc,\nc}$}; \node[] (nc_sq) at (2,0) {$\P_{\nc,\sq}$};   \node[] (nc_q) at (4,0) {$\P_{\nc,q}$};  \node[] (nc_wq) at (6,0) {$\P_{\nc,\wq}$};  \node[] (nc_c) at (8,0) {$\P_{\nc,c}$};
        
        \draw[ ->] (c_nc) to (c_sq);\draw[ ->] (c_sq) to (c_q);\draw[ ->] (c_q) to (c_wq);\draw[ ->] (c_wq) to (c_c);
        \draw[ ->] (wq_nc) to (wq_sq);\draw[ ->] (wq_sq) to (wq_q);\draw[ ->] (wq_q) to (wq_wq);\draw[ ->] (wq_wq) to (wq_c);
        \draw[ ->] (q_nc) to (q_sq);\draw[ ->] (q_sq) to (q_q);\draw[ ->] (q_q) to (q_wq);\draw[ ->] (q_wq) to (q_c);
        \draw[ ->] (sq_nc) to (sq_sq);\draw[ ->] (sq_sq) to (sq_q);\draw[ ->] (sq_q) to (sq_wq);\draw[ ->] (sq_wq) to (sq_c);
        \draw[ ->] (nc_nc) to (nc_sq);\draw[ ->] (nc_sq) to (nc_q);\draw[ ->] (nc_q) to (nc_wq);\draw[ ->] (nc_wq) to (nc_c);
        
        \draw[ ->] (c_nc) to (wq_nc);\draw[ ->] (c_sq) to (wq_sq);\draw[ ->] (c_q) to (wq_q);\draw[ ->] (c_wq) to (wq_wq);\draw[ ->] (c_c) to (wq_c);
        \draw[ ->] (wq_nc) to (q_nc);\draw[ ->] (wq_sq) to (q_sq);\draw[ ->] (wq_q) to (q_q);\draw[ ->] (wq_wq) to (q_wq);\draw[ ->] (wq_c) to (q_c);
        \draw[ ->] (q_nc) to (sq_nc);\draw[ ->] (q_sq) to (sq_sq);\draw[ ->] (q_q) to (sq_q);\draw[ ->] (q_wq) to (sq_wq);\draw[ ->] (q_c) to (sq_c);
        \draw[ ->] (sq_nc) to (nc_nc);\draw[ ->] (sq_sq) to (nc_sq);\draw[ ->] (sq_q) to (nc_q);\draw[ ->] (sq_wq) to (nc_wq);\draw[ ->] (sq_c) to (nc_c);
    \end{tikzpicture}
    \caption{All the trivial reductions}
\label{fig:expanded-trivial-reductions}
\end{figure}
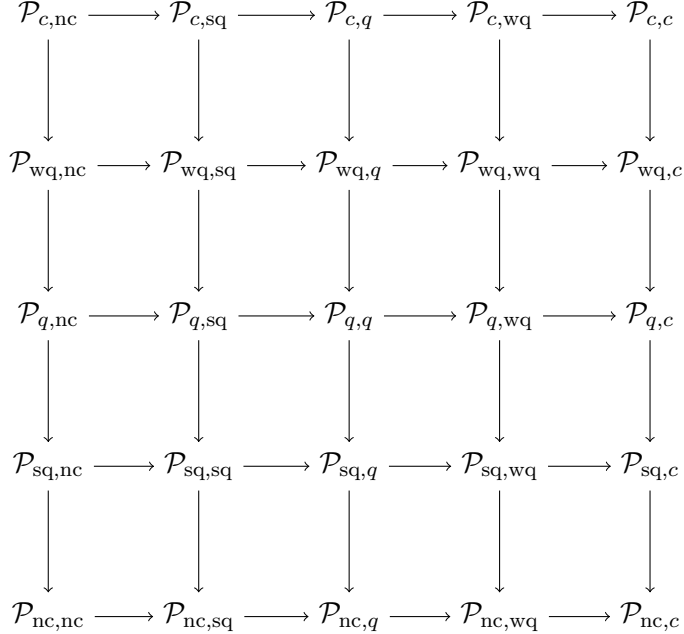

\subsection{Label-Cover and Unique-Label-Cover}\label{sec:label-cover}
The Label-Cover problem over the alphabet of size $m$, denoted $\LC(m)$, is a subset of $\CSP{2}(m)$ defined below. When it is clear from the context we drop the alphabet size and simply write $\LC$.
\begin{definition}[Label-Cover]\label{def:label-cover}
    An instance in $\LC(m)$ consists of 
\begin{itemize}
\item a bipartite graph $G = ((U,V),E)$ where $E \subseteq U \times V$,
\item a map $\pi_{u,v}:[m]\rightarrow [m]$ for every edge $(u,v)\in E$, and 
\item a weight $p_{u,v}$ for every edge $(u,v) \in E$ that defines a probability distribution.
\end{itemize}
We denote this instance with the tuple $((U,V),E,\pi,p)$.

To view this as an instance of $\CSP{2}(m)$ with variable set $U \cup V$ and constraint set $E$, as given in Definition \ref{def:general-csp}, it just remains to define the predicate $f_e:[m]^2 \rightarrow \{0,1\}$ for each edge. For $e = (u,v)\in E$, we define $f_e$ so that $f_e(a,b)= 1$ if and only if $\pi_{u,v}(a)=b$.
\end{definition}
 
The maps $\pi_{u,v}$ in the above definition are often referred to as projections. Any label assigned to the vertex $u\in U$ uniquely determines, via $\pi_{u,v}$, the label that should be assigned to its neighbor $v\in V$ for the corresponding edge predicate to evaluate to one.

We now summarize the important theorems known in the literature about Label-Cover.
\begin{theorem}[PCP theorem]\label{thm:pcp-theorem}
    For all $\delta>0$, there is a large enough alphabet over which $\LC_c(1,\delta)$ is $\NP$-hard. 
\end{theorem}

\begin{theorem}[$\MIP^* = \RE$~\cite{ji_mip_re}]\label{thm:our-label-cover-is-hard}
    For all $\delta>0$, there is a large enough alphabet over which $\LC_q(1,\delta)$ is $\RE$-hard.
\end{theorem}
\begin{proof}
In fact~\cite{ji_mip_re} proves a stronger statement: For all $\delta>0$, there is a large enough alphabet size over which $\LC_{q,\nc}(1,\delta)$ is $\RE$-hard. The theorem now follows noting the trivial reduction $\LC_{q,\nc} \longrightarrow \LC_{q}$ (see Figure \ref{fig:expanded-trivial-reductions}).
\end{proof}

\begin{definition}[Unique-Label-Cover]\label{def:unique-label-cover}
The Unique-Label-Cover problem with alphabet size $m$, denoted $\ULC(m)$, is a subset of $\LC(m)$ where projection maps $\pi_{u,v}$ are bijective. In a Unique-Label-Cover instance, for every edge $(u,v)$, the inverse $\pi_{u,v}^{-1}$ also exists, and we often use the notation $\pi_{v,u} \coloneqq \pi_{u,v}^{-1}$.
\end{definition}
In the case of a Unique-Label-Cover any label assigned to the vertex $u\in U$ uniquely determines, via $\pi_{u,v}$, the label that should be assigned to its neighbor $v\in V$ for the corresponding edge predicate to evaluate to one, and vice versa.
\begin{conjecture}[Unique Games Conjecture (UGC)~\cite{khot_original}]\label{conj:ugc}
    For all $\zeta,\delta>0$, there is a large enough alphabet over which $\ULC_c(1-\zeta,\delta)$ is $\NP$-hard.
\end{conjecture}

\begin{theorem}[KRT~\cite{kempe}]\label{thm:kempe} For every $\zeta,\delta > 0$ such that $6\zeta +\delta < 1$, the decision problem $\ULC_\nc(1-\zeta,\delta)$ can be solved in polynomial time.
\end{theorem}
\begin{proof} The KRT algorithm is originally stated for the tensor-product value of the nonlocal game instance associated with the Unique-Label-Cover instance. It is straightforward to modify their theorems to obtain an algorithm, with the same guarantee, for the quantum synchronous value of the nonlocal game. Finally Section 10.2 in~\cite{original}, shows the correspondence between the quantum synchronous value of a nonlocal game and the noncommutative value of a CSP. See also Table \ref{tab:csp-two-provers}. 

We now restate the KRT algorithm in our terminology. Given an instance $\phi \in \ULC$ with $\omega_\nc(\phi) \geq 1 - \zeta$, the KRT algorithm constructs a noncommutative assignment with value at least $1-6\zeta$. Thus whenever $1 - 6\zeta > \delta$, the decision problem $\ULC_\nc(1-\zeta,\delta)$ can be decided in polynomial time.
\end{proof}

KRT algorithm does not say anything about the quantum value as the assignment generated by this algorithm is highly noncommutative. So we propose the following conjecture: 
\begin{conjecture}[Quantum Unique Games Conjecture (qUGC)]\label{conj:qugc}
    For all $\zeta,\delta>0$, there is a large enough alphabet over which $\ULC_q(1-\zeta,\delta)$ is $\RE$-hard.
\end{conjecture}

We now define the weak-quantum assignment for $\LC(m)$ instances. Since this is a special case of quantum assignment, we will first repeat the definition of quantum assignment specifically tailored for $\LC(m)$ for the reader's convenience.
\begin{definition}[Specialization of Definition \ref{def:quantum-value} for Label-Cover] \label{def:quantumval-specLC}
A quantum assignment to an instance $\phi = ((U,V),E,\pi,p)$ in $\LC(m)$ consists of a Hilbert space $\mathcal{H}$ and a function $\Pi:U\cup V \rightarrow \PVM_m(\mathcal{H})$ such that if $(u,v)$ is an edge, then $\Pi_u$ and $\Pi_v$ must be simultaneously measurable.
\end{definition}

\begin{definition}[Weak-quantum assignment and weak-quantum value]\label{def:weak-quantum-value}
Let $(\mathcal{H},\Pi)$ be a quantum assignment to an instance $\phi = ((U,V),E,\pi,p)$ in $\LC(m)$. We say that this is a weak-quantum assignment if it satisfies some additional simultaneous measurability constraints: whenever $v,v'\in V$ share a neighbor in $U$ their assigned PVMs $\Pi_{v}$ and $\Pi_{v'}$ must be simultaneously measurable.  

The value of the assignment on the edge $(u,v)\in E$ is
$$\sum_{a \in [m]} \tr\Paren{\Pi_{u}^{a}\Pi_{v}^{\pi_{u,v}(a)}},$$ and the value of the assignment on the instance, denoted $\omega(\phi,\Pi)$, is the weighted sum $$\sum_{(u,v)\in E} p_{u,v} \sum_{a \in [m]} \tr\Paren{\Pi_{u}^{a} \Pi_{v}^{\pi_{u,v}(a)}}.$$ Finally, the weak-quantum value of the instance $\omega_{\wq}(\phi)$ is the supremum of $\omega(\phi,\Pi)$ over all weak-quantum assignments. We have the chain of inequalities $\omega_\nc(\phi) \geq \omega_q(\phi) \geq \omega_\wq(\phi) \geq \omega_c(\phi).$ 
\end{definition}

Similar to Conjecture \ref{conj:qugc} concerning the quantum value of Unique-Label-Cover, we can propose a corresponding conjecture regarding the weak-quantum value:
\begin{conjecture}[Weak-Quantum Unique Games Conjecture (wqUGC)]\label{conj:wqugc}
    For all $\zeta,\delta>0$, there is a large enough alphabet over which $\ULC_{\wq}(1-\zeta,\delta)$ is $\RE$-hard.
\end{conjecture}
\begin{figure}[ht]
\centering
	\begin{tikzpicture}

        \node[] (c_nc) at (0,8) {$\ULC_{c,\nc}$}; \node[] (c_sq) at (2.5,8) {$\ULC_{c,\sq}$};   \node[] (c_q) at (5,8) {$\ULC_{c,q}$};  \node[] (c_wq) at (7.5,8) {$\ULC_{c,\wq}$};  \node[red] (c_c) at (10,8) {$\ULC_{c,c}$};
        \node[] (wq_nc) at (0,6) {$\ULC_{\wq,\nc}$}; \node[] (wq_sq) at (2.5,6) {$\ULC_{\wq,\sq}$};   \node[] (wq_q) at (5,6) {$\ULC_{\wq,q}$};  \node[red] (wq_wq) at (7.5,6) {$\ULC_{\wq,\wq}$};  \node[] (wq_c) at (10,6) {$\ULC_{\wq,c}$};
        \node[] (q_nc) at (0,4) {$\ULC_{q,\nc}$}; \node[] (q_sq) at (2.5,4) {$\ULC_{q,\sq}$};   \node[red] (q_q) at (5,4) {$\ULC_{q,q}$};  \node[] (q_wq) at (7.5,4) {$\ULC_{q,\wq}$};  \node[] (q_c) at (10,4) {$\ULC_{q,c}$};
        \node[] (sq_nc) at (0,2) {$\ULC_{\sq,\nc}$}; \node[] (sq_sq) at (2.5,2) {$\ULC_{\sq,\sq}$};   \node[] (sq_q) at (5,2) {$\ULC_{\sq,q}$};  \node[] (sq_wq) at (7.5,2) {$\ULC_{\sq,\wq}$};  \node[] (sq_c) at (10,2) {$\ULC_{\sq,c}$};
        \node[blue] (nc_nc) at (0,0) {$\ULC_{\nc,\nc}$}; \node[] (nc_sq) at (2.5,0) {$\ULC_{\nc,\sq}$};   \node[] (nc_q) at (5,0) {$\ULC_{\nc,q}$};  \node[] (nc_wq) at (7.5,0) {$\ULC_{\nc,\wq}$};  \node[] (nc_c) at (10,0) {$\ULC_{\nc,c}$};
        
        \draw[ ->] (c_nc) to (c_sq);\draw[ ->] (c_sq) to (c_q);\draw[ ->] (c_q) to (c_wq);\draw[ ->] (c_wq) to (c_c);
        \draw[ ->] (wq_nc) to (wq_sq);\draw[ ->] (wq_sq) to (wq_q);\draw[ ->] (wq_q) to (wq_wq);\draw[ ->] (wq_wq) to (wq_c);
        \draw[ ->] (q_nc) to (q_sq);\draw[ ->] (q_sq) to (q_q);\draw[ ->] (q_q) to (q_wq);\draw[ ->] (q_wq) to (q_c);
        \draw[ ->] (sq_nc) to (sq_sq);\draw[ ->] (sq_sq) to (sq_q);\draw[ ->] (sq_q) to (sq_wq);\draw[ ->] (sq_wq) to (sq_c);
        \draw[ ->] (nc_nc) to (nc_sq);\draw[ ->] (nc_sq) to (nc_q);\draw[ ->] (nc_q) to (nc_wq);\draw[ ->] (nc_wq) to (nc_c);
        
        \draw[ ->] (c_nc) to (wq_nc);\draw[ ->] (c_sq) to (wq_sq);\draw[ ->] (c_q) to (wq_q);\draw[ ->] (c_wq) to (wq_wq);\draw[ ->] (c_c) to (wq_c);
        \draw[ ->] (wq_nc) to (q_nc);\draw[ ->] (wq_sq) to (q_sq);\draw[ ->] (wq_q) to (q_q);\draw[ ->] (wq_wq) to (q_wq);\draw[ ->] (wq_c) to (q_c);
        \draw[ ->] (q_nc) to (sq_nc);\draw[ ->] (q_sq) to (sq_sq);\draw[ ->] (q_q) to (sq_q);\draw[ ->] (q_wq) to (sq_wq);\draw[ ->] (q_c) to (sq_c);
        \draw[ ->] (sq_nc) to (nc_nc);\draw[ ->] (sq_sq) to (nc_sq);\draw[ ->] (sq_q) to (nc_q);\draw[ ->] (sq_wq) to (nc_wq);\draw[ ->] (sq_c) to (nc_c);
    \end{tikzpicture}
    \caption{All the trivial reductions between decision variants of Unique-Label-Cover are illustrated. The decision variants of Unique-Label-Cover that concern each of the UGC, qUGC, and wqUGC are indicated in red. By the KRT algorithm, the decision problem $\ULC_{\nc,\nc}$ is easy. See also Figure \ref{fig:unique-label-cover-final-complexity-landscape}.}
\label{fig:expanded-reductions-unique-games}
\end{figure}
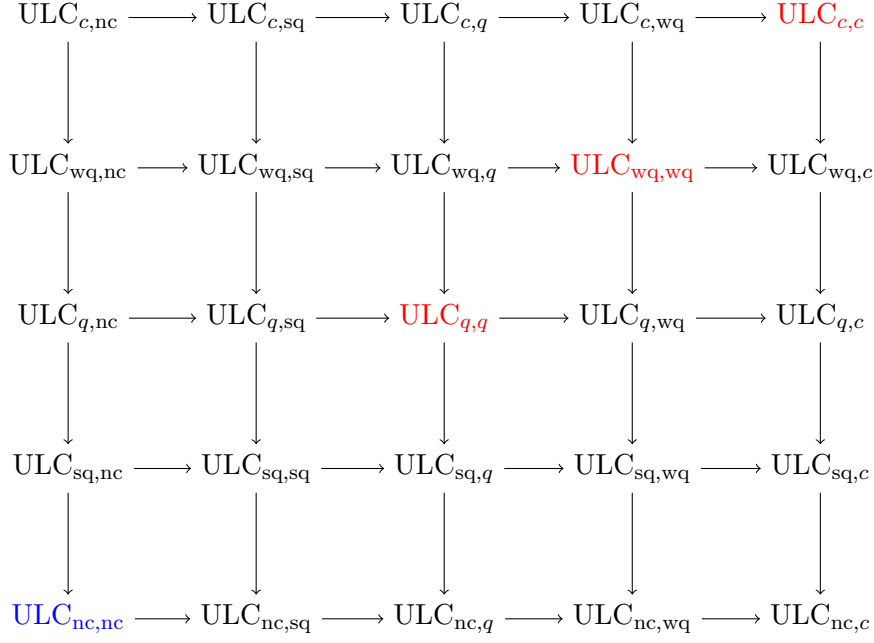

The qUGC and wqUGC are two variants of the quantum Unique Games Conjecture, but there are many others. For every $s,t \in \{c,\wq,q,\sq,\nc\}$, let UGC$_{s,t}$ denote the assumption that: ``For every $\zeta,\delta > 0$, there exists a sufficiently large alphabet over which $\ULC_{s,t}(1-\zeta,\delta)$ is $\RE$-hard.'' 

What is the relation between qUGC and wqUGC? At present, we do not know which conjecture represents a stronger assumption. As shown in Figure \ref{fig:expanded-reductions-unique-games}, no reduction is currently known between $\ULC_q$ and $\ULC_\wq$ in either direction.\footnote{Establishing a reduction from $\ULC_q$ to $\ULC_\wq$ faces challenges. Informally, proving completeness for any such reduction involves determining whether, given a set of nearly commuting measurements, there always exists a set of perfectly commuting measurements that can simulate the behavior of the original set. These questions are closely tied to problems in the theory of group stability~\cite{ioana2021commutingmatrices,chapman,delasalle}, which seem at the moment out of reach.} 

As illustrated in Figure \ref{fig:strongest-conjecture}, the UGC$_{\wq,q}$ is a stronger assumption than both qUGC and wqUGC.
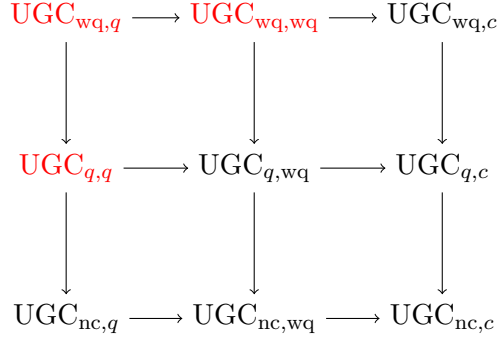
\begin{figure}[ht]
\centering
	\begin{tikzpicture}

       \node[red] (wq_q) at (5,4) {$\UGC_{\wq,q}$};  \node[red] (wq_wq) at (7.5,4) {$\UGC_{\wq,\wq}$};  \node[] (wq_c) at (10,4) {$\UGC_{\wq,c}$};
          \node[red] (q_q) at (5,2) {$\UGC_{q,q}$};  \node[] (q_wq) at (7.5,2) {$\UGC_{q,\wq}$};  \node[] (q_c) at (10,2) {$\UGC_{q,c}$};
          \node[] (nc_q) at (5,0) {$\UGC_{\nc,q}$};  \node[] (nc_wq) at (7.5,0) {$\UGC_{\nc,\wq}$};  \node[] (nc_c) at (10,0) {$\UGC_{\nc,c}$};

        \draw[ ->] (wq_q) to (wq_wq);\draw[ ->] (wq_wq) to (wq_c);
        \draw[ ->] (q_q) to (q_wq);\draw[ ->] (q_wq) to (q_c);
        \draw[ ->] (nc_q) to (nc_wq);\draw[ ->] (nc_wq) to (nc_c);

        \draw[ ->] (wq_q) to (q_q);\draw[ ->] (wq_wq) to (q_wq);\draw[ ->] (wq_c) to (q_c);
        \draw[ ->] (q_q) to (nc_q);\draw[ ->] (q_wq) to (nc_wq);\draw[ ->] (q_c) to (nc_c);
    \end{tikzpicture}
    \caption{Variants of quantum Unique Games Conjectures: $\UGC_{\wq,q}$ is the strongest conjecture and it implies both qUGC and wqUGC.}
\label{fig:strongest-conjecture}
\end{figure} 

In this paper we use both qUGC and wqUGC. The qUGC is used in proving hardness for $\Lin{2}$ in Section \ref{sec:2lin}. The wqUGC is used in proving hardness for MaxCut in Section \ref{sec:maxcut}.

\subsection{$k$-Lin and MaxCut}\label{sec:definitions-lin-maxcut}
The $\Lin{k}(2)$ is a subset of $\CSP{k}(2)$ where the constraints are linear equations over $\Z_2$. Since all our linear CSPs are defined over $\Z_2$, we simply write $\Lin{k} \coloneq \Lin{k}(2)$. We now formally define an instance in $\Lin{k}$.
\begin{definition}\label{def:linear-constraints}
    An instance in $\Lin{k}$ is a tuple $(V,E,r,p)$ where
\begin{itemize}
\item $V$ is some set and $E$ is a multiset of $k$-tuples of $V$,
\item $r_{e} \in \{0,1\}$ is a parity bit for every $e\in E$, 
\item and $p$ is a probability distribution  on $E$.
\end{itemize}

To view this as an instance of $\CSP{k}(2)$, as in Definition \ref{def:general-csp}, we define predicates $f_{e}:\{0,1\}^k \rightarrow \{0,1\}$ for every edge $e\in E$ such that for every $a \in \{0,1\}^k$, $f_{e}(a_1,\ldots,a_k)=1$ if and only if $\sum_{i\in [k]} a_i = r_{e}$, where the arithmetic is over $\Z_2$.
\end{definition}
The special case $\Lin{3}$ is sometimes referred to as $\XOR$. The set of all instances in $\Lin{2}$ where all the parity bits are $1$ is known as MaxCut. Since the parity bits are fixed for MaxCut, we denote its instances by triples $(V,E,p)$.

\begin{example*}
Consider the following weighted linear system over $\Z_2$ 
\begin{align*}
    &x_1 + x_2 = 0,\\
    &x_1 + x_2 = 0,\\
    &x_1 + x_2 = 1,\\
    &x_3 + x_4 = 1,
\end{align*}
with weights $1/2,1/4,1/8,1/8$. An instance $\psi=(V,E,r,p)$ of $\Lin{2}$ that represents this system is given by $V = \{1,2,3,4\}, E = \{e_1,e_2,e_3,e_4\}$, where $e_1 = (1,2), e_2 = (1,2), e_3 = (1,2), e_4 = (3,4)$, with parity bits $r_{e_1} = 0,r_{e_2} = 0,r_{e_3} = 1,r_{e_4} = 1$ and edge weights $p_{e_1} = 1/2,p_{e_2} = 1/4,p_{e_3} = 1/8,p_{e_4} = 1/8$. Note that there is some flexibility when converting a linear system to an instance of $\Lin{2}$; for instance, one could let $e_4 = (4, 3)$ instead of $e_4 = (3, 4)$. 
\end{example*}

For binary-alphabet CSPs, it is often more convenient to use the multiplicative group of $\{\pm1\}$ rather than $\Z_2$ as the set of labels. 

\begin{example*}
The linear system from the previous example can also be written multiplicatively as
\begin{align*}
    &x_1 x_2 = 1,\\
    &x_1 x_2 = 1,\\
    &x_1 x_2 = -1,\\
    &x_3 x_4 = -1,
\end{align*}
where $x_1,x_2,x_3,x_4 \in \{\pm 1\}$. To reflect this in the corresponding instance $\psi=(V,E,r,p)$ we set $r_{e_1} = 1,r_{e_2} = 1,r_{e_3} = -1,r_{e_4} = -1$.
\end{example*}
From here on, we assume that $\Lin{2}$ instances are multiplicative, meaning that their parity bits are $r_e \in \{\pm1\}$.

With this change, it becomes more natural to write assignments to $\Lin{2}$ instances in multiplicative form as well. The multiplicative analogue of binary-outcome PVMs are observables. Recall from section \ref{sec:notations}, that binary-outcome PVMs correspond one-to-one with observables. For convenience, we will restate the definition of quantum assignments in terms of observables.

\begin{definition}[Quantum assignment and quantum value for $\Lin{k}$ stated in terms of observables]\label{def:quantum-observable-assignment}
A quantum assignment to a $\Lin{k}$ instance $(V,E,r,p)$ consists of a Hilbert space $\mathcal{H}$ and a function $\alpha:V \to \Obs(\mathcal{H})$ such that for every constraint $e=(i_1,\ldots,i_k)$ the observables $\alpha_{i_1},\ldots,\alpha_{i_k}$ commute. 

The value of the assignment $\alpha$ on the constraint $e=(i_1,\ldots,i_k)$ is $$\frac{1}{2} + \frac{1}{2}r_{e} \tr(\alpha_{i_1} \cdots \alpha_{i_k}).$$ 

The value of the assignment $\alpha$ on the instance, denoted $\omega(\phi,\alpha)$, is the weighted sum $$\frac{1}{2} + \frac{1}{2}\sum_{e=(i_1,\ldots,i_k)\in E} p_e\, r_{e} \tr(\alpha_{i_1} \cdots \alpha_{i_k}).$$ 

Finally, the quantum value of the instance is the supremum of $\omega(\phi,\alpha)$ over all quantum assignments $\alpha$. 
\end{definition}

\section{$2$-LIN}\label{sec:2lin}
\subsection{Statements}
In this section, we present a proof of the following theorem:

\begin{theorem}\label{thm:lin2hardness}
    Assuming qUGC, for every $\frac{1}{2}<t<1$ and sufficiently small $\varepsilon>0$, $\Lin{2}_q(1-O(\varepsilon), 1-O(\varepsilon^t))$ is $\RE$-hard. 
\end{theorem}
This result should be compared to the following theorem by Khot~\cite{khot_original}:
\begin{theorem}\label{thm:khot-lin2hardness}
    Assuming UGC, for every $\frac{1}{2}<t<1$ and sufficiently small $\varepsilon>0$, $\Lin{2}_c(1-O(\varepsilon), 1-O(\varepsilon^t))$ is $\NP$-hard. 
\end{theorem}
Our proof is an extension of Khot's proof to the operator setting (according to Khot \cite{khot_original}, the overall template of the  proof is due to an unpublished work by H\r{a}stad).

Theorem \ref{thm:lin2hardness} follows directly as a corollary of the next theorem:
\begin{theorem}\label{thm:lin2reduction}
    For every $\frac{1}{2}<t<1$ and sufficiently small $\varepsilon>0$, there exists $\zeta,\delta>0$ such that
    $$\ULC_q(1-\zeta,\delta) \longrightarrow \Lin{2}_q(1-O(\varepsilon), 1-O(\varepsilon^t)).$$
\end{theorem}
To prove Theorem \ref{thm:lin2reduction}, we will first describe the reduction and then establish its completeness and soundness in the following sections. 

\,

\noindent\textbf{Reduction.}
    Fix $\frac{1}{2}<t<1$ and $\varepsilon>0$. Let $\zeta,\delta>0$ be constants to be determined later. 
    
    Given an instance $\phi=((U,V),E, \pi, p)$ in $\ULC$, we now describe how to construct an instance $\psi=(\oV,\oE,\orr,\opp)$ in $\Lin{2}$. Let $m$ be the alphabet size of $\phi$.

    \,
    
    \textbf{Distribution.} We first define a probability distribution. Let $\Pr:U\times V\times \{\pm1\}^m\times \{\pm1\}^m \to [0,1]$ be given by \[\Pr(u,v,x,\mu) = \frac{p_{u,v}}{2^m}\varepsilon^{|\mu|}(1-\varepsilon)^{m-|\mu|}\] where $|\mu|$, the Hamming weight, denotes the number of $-1$'s in $\mu$. This clearly defines a valid distribution. Indeed $\Pr(u,v,x,\mu)$ is the probability of independently sampling
    \begin{itemize}
        \item $(u,v)$ from the distribution $p$, 
        \item $x$ from the uniform distribution on $\{\pm1\}^m$, and 
        \item $\mu$ from the noise distribution $\varepsilon^{|\mu|}(1-\varepsilon)^{m-|\mu|}$.
    \end{itemize}  
    We also use $\Pr$ to denote the marginal distributions when no ambiguity arises. For example, for every $u \in U,x,\mu\in \{\pm1\}^m$, we write $\Pr(u,x,\mu)$ to mean $$\sum_{v\in V} \Pr(u,v,x,\mu).$$ Similarly, we have $$\Pr(u) = \sum_{v,x,\mu} \Pr(u,v,x,\mu) = \sum_{v: (u,v)\in E} p_{u,v}$$ for every $u \in U$. Lastly, whenever we write $\mathbb{E}_{u, v, x, \mu}$ (or with any other subset of the indices $u, v, x, \mu$), the expectation is taken with respect to the distribution $\Pr$.

    \,
    
    \textbf{Formula for the Value of $\boldsymbol{\phi.}$} For example, given a quantum assignment $\Pi:U\cup V \to \Obs(\mathcal{H})$ to $\phi$, its value is given by the following expectation
    \begin{equation}\label{eq:2lin-reduction-value-of-ulc}
        \omega(\phi,\Pi) = \sum_{(u,v) \in E} p_{u,v} \sum_{a\in [m]} \tr\Bigparen{\Pi_u^a \Pi_v^{\pi_{u,v}(a)}} = \expect_{u,v}\Brac{\sum_{a\in [m]} \tr\Bigparen{\Pi_u^a \Pi_v^{\pi_{u,v}(a)}}}.
    \end{equation}

    \,

    \textbf{Construction of $\boldsymbol{\psi}$.} We now construct $\psi = (\oV,\oE,\orr,\opp)$. The vertex set is defined as $\oV = (U \cup V) \times \{\pm1\}^m$. We define three types of constraints:
    \begin{itemize}
        \item $\oE_1$ consists of $((u,x),(u,x\mu))$ for all $u \in U$ and $x,\mu \in \{\pm1\}^m$,
        \item $\oE_2$ consists of $((v,x),(v,x\mu))$ for all $v \in U$ and $x,\mu \in \{\pm1\}^m$, and
        \item $\oE_3$ consists of $((u,x),(v,(x\circ \pi_{v,u})\mu))$ for all $(u,v) \in E$ and $x,\mu \in \{\pm1\}^m$,
    \end{itemize}
    and set $\oE = \oE_1\cup \oE_2 \cup \oE_3.$ All parity bits $\orr_e$ are set to $1$. For every $u,v,x,\mu$, the edge weights are defined as follows:
    \begin{align*}
        \opp_{(u,x),(u,x\mu)} = \frac{1}{4}\Pr(u,x,\mu), \quad
        \opp_{(v,x),(v,x\mu)} = \frac{1}{4}\Pr(v,x,\mu), \quad
        \opp_{(u,x),(v,(x\circ \pi_{v,u})\mu)} = \frac{1}{2}\Pr(u,v,x,\mu).
    \end{align*}
    It should be clear that this defines a probability distribution on $\oE$, and this completes the construction of $\psi$. 
    
    There is a minor issue with this construction, which will be addressed at the end of this section using a technique known as \emph{folding}. For now, as we did with the instance $\phi$, we will derive a formula for the quantum value of $\psi$.

    \,
    
    \textbf{Formula for the Value of $\boldsymbol{\psi.}$} Let $\alpha:\oV\to\Obs(\mathcal{H})$ be a quantum assignment to $\psi$. We use the shorthand notation $\alpha_u(x) \coloneqq \alpha((u,x))$ for all $(u,x)\in \oV$. Using the general formula from Definition \ref{def:quantum-observable-assignment} and the fact that $\orr_e = 1$ we obtain
    \[\omega(\psi,\alpha) = \frac{1}{2} + \frac{1}{2}\sum_{e=(w_1,w_2) \in \oE}\opp_e \tr(\alpha(w_1)\alpha(w_2)).\]
    Breaking down the sum according to the edge types and expanding the edge weights, we obtain:
    \begin{align*}
        \sum_{e=(w_1,w_2) \in \oE}\opp_e \tr(\alpha(w_1)\alpha(w_2)) &= \frac{1}{4}\sum_{u,x,\mu} \Pr(u,x,\mu) \tr\Bigparen{\alpha_u(x)\alpha_u(x\mu)}\\
                                                                  &+ \frac{1}{4}\sum_{v,x,\mu} \Pr(v,x,\mu) \tr\Bigparen{\alpha_v(x)\alpha_v(x\mu)}\\
                                                                  &+ \frac{1}{2}\sum_{u,v,x,\mu}\Pr(u,v,x,\mu) \tr\Bigparen{\alpha_u(x)\alpha_v((x\circ\pi_{v,u})\mu)}
    \end{align*}
    which can be rewritten as:
    \begin{align*}
        \expect_{u,v,x,\mu}\Brac{\tr\Bigparen{ \frac{1}{4}\alpha_u(x)\alpha_u(x\mu)
                                            +\frac{1}{4}\alpha_v(x)\alpha_v(x\mu)
                                            + \frac{1}{2}\alpha_u(x)\alpha_v((x\circ\pi_{v,u})\mu)}}.
    \end{align*}
    Therefore, the value of the assignment is
    \begin{equation}\label{eq:valueLin_reduction}
        \omega(\psi,\alpha) = \frac{1}{2} + \frac{1}{8}\expect_{u,v,x,\mu}\Brac{\tr\Bigparen{\alpha_u(x)\alpha_u(x\mu)
                                                                                        +\alpha_v(x)\alpha_v(x\mu)
                                                                                    + 2\alpha_u(x)\alpha_v((x\circ\pi_{v,u})\mu)}}.
    \end{equation}

\subsubsection{Folding}\label{sec:folding} 
As it stands, there is an issue with our construction of $\psi$. Since all the parity bits in $\psi$ are set to $1$, the instance becomes trivial: even the constant classical assignment achieves the maximum classical value of $1$. However, this is not a problem because we are actually interested in the \emph{folded value} of $\psi$ (as opposed to its ordinary value):

    \begin{definition}[Folded assignments and values] A (classical, quantum, or noncommutative) assignment $\alpha$ to the instance $\psi$ (constructed above) is folded if $\alpha_u(-x) = -\alpha_u(x)$ for all $u \in U\cup V$ and $x \in \{\pm1\}^m$. 
    
    The folded (classical, quantum, or noncommutative) value of the instance is the maximum value over all folded (classical, quantum, or noncommutative) assignments. 
    \end{definition}
    For example, the constant assignment is not folded. The \emph{folding trick}, which we discuss next, shows how to construct an instance $\psi'$ from $\psi$ such that the folded value of $\psi$ is the same as the ordinary value of $\psi'$. This is a standard trick from the PCP literature \cite{hastad1}, and the following lemma serves as the basis for its use in this paper. The proof is provided in Section \ref{sec:folding-lemma}. 
    \begin{lemma}[Folding Trick]\label{lem:folding}
    Given an instance $\psi$ as constructed above, there exists another efficiently constructed instance $\psi' \in \Lin{2}$ such that assignments to $\psi'$ one-to-one correspond with folded assignments to $\psi$. Furthermore, this correspondence preserves the value. Consequently, the folded value of $\psi$ is the same as the ordinary value of $\psi'$.
    \end{lemma}
    Therefore, for the remainder of the proof of Theorem \ref{thm:lin2reduction}, we can assume, without loss of generality, that valid assignments to $\psi$ are folded.

\subsection{Completeness}
Let $\Pi:U\cup V \to \PVM_m(\hilb)$ be a quantum assignment to $\phi$ such that $\omega(\phi,\Pi)\geq 1-\zeta$. We construct a quantum assignment $\alpha: \oV \to \Obs(\hilb)$ to $\psi$ such that $\omega(\psi,\alpha) \geq 1 - \varepsilon - \zeta/4$. Choosing $\zeta$ sufficiently small, we can ensure $\omega(\psi,\alpha) \geq 1 - 2\varepsilon$. This will complete the proof of completeness.

For every $u \in U\cup V$ and $x\in \{\pm1\}^m$, define
\begin{align*}
    \alpha_u(x) = \sum_{a\in [m]} x_a \Pi_u^a.
\end{align*}
These operators are Hermitian \[\alpha_u(x)^* = \sum_{a\in [m]} x_a^* (\Pi_u^a)^* = \sum_{a\in [m]} x_a \Pi_u^a = \alpha_u(x)\] and unitary
\[\alpha_u(x)^*\alpha_u(x) = \sum_{a,b\in [m]} x_ax_b \Pi_u^a\Pi_u^b =  \sum_{a\in [m]} \Pi_u^a = I.\]
Also functions $\alpha_u$ are odd $\alpha_u(-x) = -\alpha_u(x)$.

Furthermore, for every $u\in U\cup V$ and $x,y\in \{ \pm 1\}^m$ we have:
    \begin{align*}
    \Brac{\alpha_u(x), \alpha_u(y)} &= \sum_{a,b} x_a y_b \Brac{\Pi_u^a,\Pi_u^b}=0,
    \end{align*}
as operators in a PVM  commute. Finally for every $(u,v)\in E$ and $x,y\in \{\pm 1\}^m$ we have:
    \begin{align*}
        \Brac{\alpha_u(x), \alpha_v(y)} &= \sum_{a,b} x_a y_b \Brac{\Pi_u^a,\Pi_v^b}=0,
    \end{align*}
as $\Pi_u$ and $\Pi_v$ are simultaneously measurable when $(u,v)\in E$. Therefore $\alpha: \oV \to \Obs(\hilb)$ is a folded quantum assignment to $\psi$.

We now calculate $\omega(\psi,\alpha)$ using Equation \eqref{eq:valueLin_reduction}. First note
\begin{align*}
    \expect_{u, x,\mu}\Brac{\tr\Bigparen{\alpha_u(x) \alpha_u(x\mu)}}&=\expect_{u, x,\mu} \Brac{\sum_{a,b} x_a x_b \mu_b\tr\Bigparen{\Pi_u^a \Pi_u^b}}\\
    &= (1-2\varepsilon) \expect_{u} \Brac{\sum_{a}\tr\Bigparen{\Pi_u^a}}=1-2\varepsilon,
\end{align*}
where in the second line we used $\expect_x \Brac{x_a x_b} =\delta_{a,b}$ and $\expect_\mu \Brac{\mu_b} = 1-2\varepsilon$. Similarly note
\begin{align*}
    \expect_{u,v, x,\mu} \Brac{\tr \Bigparen{\alpha_u(x) \alpha_v((x\circ \pi_{v,u}) \mu)}} &= \expect_{u,v, x,\mu}\Brac{\sum_{a,b} x_a x_{\pi_{v,u}(b)} \mu_b \tr\Bigparen{ \Pi_u^a \Pi_v^b}}\\
    &= (1-2\varepsilon) \expect_{u,v}\Brac{ \sum_a \tr\Bigparen{ \Pi_u^a \Pi_v^{\pi_{u,v}(a)}}}\\ &\geq (1-2\varepsilon)(1-\zeta),
\end{align*}
where the inequality follows from \eqref{eq:2lin-reduction-value-of-ulc} and the assumption $\omega(\phi,\Pi)\geq 1-\zeta$. Putting everything together in \eqref{eq:valueLin_reduction}, we obtain
\begin{align*}
    \omega(\psi,\alpha) = \frac{1}{2} + \frac{1}{4}\Bigparen{(1-2\varepsilon) + (1-2\varepsilon)(1-\zeta)} \geq 1 - \varepsilon - \frac{\zeta}{4},
\end{align*}
as desired.

\subsection{Soundness}
Let $\alpha: \oV \to \Obs(\hilb)$ be a folded quantum assignment to $\psi$ such that $\omega(\psi,\alpha) \geq 1 -\frac{1}{32}b_t\varepsilon^t$, where $b_t > 0$ is a constant to be determined later. Let $d$ denote the dimension of $\hilb$. We construct a quantum assignment $\Pi:U\cup V \to \PVM_m(\hilb)$ to $\phi$ such that $\omega(\phi,\Pi) \geq \Omega(\varepsilon 4^{-2\varepsilon^{-2}})$. This will complete the proof of soundness, as we can choose $\delta$ smaller than $\Omega(\varepsilon 4^{-2\varepsilon^{-2}})$.

\,

\textbf{Good Edges.} By the assumption on the value of $\alpha$ and \eqref{eq:valueLin_reduction}, we have
\[\expect_{u,v,x,\mu}\Brac{\tr\Bigparen{\alpha_u(x)\alpha_u(x\mu)
                        +\alpha_v(x)\alpha_v(x\mu)
                        + 2\alpha_u(x)\alpha_v((x\circ\pi_{v,u})\mu)}} \geq 4 - \frac{1}{4}b_t\varepsilon^t.\]
Since $\expect_{u,v,x,\mu} \Brac{\alpha_v(x)\alpha_v(x\mu)} \leq 1$, it follows that
\[\expect_{u,v,x,\mu}\Brac{\tr\Bigparen{\alpha_u(x)\alpha_u(x\mu)
                        + 2\alpha_u(x)\alpha_v((x\circ\pi_{v,u})\mu)}} \geq 3 - \frac{1}{4}b_t\varepsilon^t.\]
Therefore, by a simple averaging argument, with probability at least half over the edges $(u,v) \in E$ (of the Unique-Label-Cover instance $\phi$), we have
\begin{equation}\label{eq:good-edge}\expect_{x,\mu}\Brac{\tr\Bigparen{\alpha_u(x)\alpha_u(x\mu)
                        + 2\alpha_u(x)\alpha_v((x\circ\pi_{v,u})\mu)}} \geq 3 - \frac{1}{2}b_t\varepsilon^t.
\end{equation}
We refer to edges satisfying \eqref{eq:good-edge} as good. 

\,

\textbf{Good Classical Assignments.} Fix a good edge $\be = (\bu,\bv)$. By the commutation relations (of quantum assignment), operators $\alpha_\bu(\{\pm1\}^m)$ and $\alpha_\bv(\{\pm1\}^m)$ are diagonalizable in the same basis. Fix such a basis and write the operators in this basis. For $j \in [d]$, let $\alpha_{\be,\bu,j}(x)$ denote the $j$th diagonal entry of $\alpha_\bu(x)$, with similar notation for $\alpha_{\be,\bv,j}$. For simplicity, we refer to them as $\alpha_{\bu,j}$ and $\alpha_{\bv,j}$ when the choice of $\be=(\bu,\bv)$ is clear from the context. 

Note that the range of $\alpha_{\bu,j}$ and $\alpha_{\bv,j}$ is $\{\pm1\}$. Thus, we are effectively extracting classical assignments from the quantum assignment.

By the definition of the dimension-normalized trace, we have 
\begin{align*}
    \tr\Bigparen{\alpha_\bu(x)\alpha_\bu(x\mu)
                        + 2\alpha_\bu(x)\alpha_{\bv}((x\circ\pi_{\bv,\bu})\mu)}= \expect_{j}\Bigbrac{\alpha_{\bu,j}(x)\alpha_{\bu,j}(x\mu)
                        + 2\alpha_{\bu,j}(x)\alpha_{\bv,j}((x\circ\pi_{\bv,\bu})\mu)}
\end{align*}
where $\expect_j$ is to be taken with respect to the uniform distribution on $[d]$. Thus,
\begin{align*}
    \expect_{j}\Bigbrac{\expect_{x,\mu}\Bigbrac{\alpha_{\bu,j}(x)\alpha_{\bu,j}(x\mu)
                        + 2\alpha_{\bu,j}(x)\alpha_{\bv,j}((x\circ\pi_{\bv,\bu})\mu)}} \geq 3 - \frac{1}{2}b_t\varepsilon^t.
\end{align*}
Again, by a simple averaging argument, for at least half of the indices $j \in [d]$, we have
\begin{align}\label{eq:good-j}
    \expect_{x,\mu}\Bigbrac{\alpha_{\bu,j}(x)\alpha_{\bu,j}(x\mu)
                        + 2\alpha_{\bu,j}(x)\alpha_{\bv,j}((x\circ\pi_{\bv,\bu})\mu)} \geq 3 - b_t\varepsilon^t.
\end{align}
We call such indices good.  

\,

\textbf{Two Inequalities.} Fix a good index $\bj$. Since both $\alpha_{\bu,\bj}(x)\alpha_{\bu,\bj}(x\mu)$ and $\alpha_{\bu,\bj}(x)\alpha_{\bv,\bj}((x\circ\pi_{\bv,\bu})\mu)$ are bounded above by $1$, from \eqref{eq:good-j}, we obtain the following two inequalities:
\begin{align*}
    \expect_{x,\mu}\Bigbrac{\alpha_{\bu,\bj}(x)\alpha_{\bu,\bj}(x\mu)} &\geq 1 - b_t\varepsilon^t,\\
    \expect_{x,\mu}\Bigbrac{\alpha_{\bu,\bj}(x)\alpha_{\bv,\bj}((x\circ\pi_{\bv,\bu})\mu)} &\geq 1 - \frac{1}{2}b_t\varepsilon^t.
\end{align*}
For ease of notation, let $\sigma = \pi_{\bu,\bv}$, $\beta = \alpha_{\bu,\bj}$, and $\gamma = \alpha_{\bv,\bj}$. Rewriting the inequalities with these shorthand notations gives:
\begin{align*}
    \expect_{x,\mu}\bigbrac{\beta(x)\beta(x\mu)} &\geq 1 - b_t\varepsilon^t,\\
    \expect_{x,\mu}\bigbrac{\beta(x)\gamma((x\circ\sigma^{-1})\mu)} &\geq 1 - \frac{1}{2}b_t\varepsilon^t.
\end{align*}
The next part of the proof derives the consequences of these two inequalities using Fourier analysis, following the template established by Khot and H\r{a}stad. 

First, let us express these inequalities in terms of the Fourier expansions:
\begin{align}
    \sum_{S}\hat{\beta}(S)^2(1-2\varepsilon)^{|S|} &\geq 1 - b_t\varepsilon^t,\label{eq:first-implication}\\
    \sum_{S}\hat{\beta}(S)\hat{\gamma}(\sigma(S))(1-2\varepsilon)^{|S|} &\geq 1 - \frac{1}{2}b_t\varepsilon^t.\label{eq:second-implication}
\end{align}
These Fourier calculations are standard.

\,

\textbf{Consequences of the First Inequality.} We now derive the consequences of \eqref{eq:first-implication}. Define two sets $\S_1,\S_2$ of subsets of $\{\pm1\}^m$:
\begin{align*}
    \S_1 = \Big\{S \subseteq \{\pm1\}^m : |S| < \varepsilon^{-1} \Big\}, \quad \S_2 = \Big\{S \subseteq \{\pm1\}^m : |\hat{\beta}(S)| > \frac{1}{10}4^{-\varepsilon^{-2}}\Big\}.
\end{align*}
Using \eqref{eq:first-implication}, we can write
\begin{align*}
    1 - b_t\varepsilon^t &\leq \sum_{S}\hat{\beta}(S)^2(1-2\varepsilon)^{|S|}\\
    &= \sum_{S \in \S_1} \hat{\beta}(S)^2(1-2\varepsilon)^{|S|} + \sum_{S \notin \S_1}  \hat{\beta}(S)^2(1-2\varepsilon)^{|S|}\\
    &\leq \sum_{S \in \S_1} \hat{\beta}(S)^2 + e^{-2}\sum_{S \notin \S_1}  \hat{\beta}(S)^2,
\end{align*}
where in the last step we used the inequality $(1-2\varepsilon)^{|S|} \leq e^{-2\varepsilon |S|}$. Applying Parseval's identity, we obtain
\begin{align}\label{eq:bourgain-premise}
    \sum_{S \notin \S_1}  \hat{\beta}(S)^2 < \frac{b_t}{1-e^{-2}}\varepsilon^t.
\end{align}
We now choose $b_t = c_t(1-e^{-2})$, where $c_t$ is the constant in Bourgain's theorem. We apply Bourgain's theorem (Theorem \ref{thm:Bourgain}) to \eqref{eq:bourgain-premise}, with the parameter $k$ in the statement of the theorem set to $\varepsilon^{-1}$. The conclusion is that 
\begin{align}\label{eq:bourgain-conclusion}
    \sum_{S \notin \S_2} \hat{\beta}(S)^2 < \frac{1}{100}.
\end{align}

\textbf{Consequence of the Second Inequality.} We now focus on \eqref{eq:second-implication} and aim to show that the contribution of $S \in \S_1 \cap \S_2$ to the left hand side is significant. Specifically, we want to prove
\begin{align}
\Big|\sum_{S \in \S_1\cap \S_2}\hat{\beta}(S)\hat{\gamma}(\sigma(S))(1-2\varepsilon)^{|S|}\Big| > \frac{1}{4}.\label{eq:counclusion-of-everything}
\end{align}

First, note that by Cauchy-Schwarz, Parseval's identity, and \eqref{eq:bourgain-premise}, we have
\begin{align}
    \Bigabs{\sum_{S\notin \S_1} \hat{\beta}(S) \hat{\gamma}(\sigma(S)) (1-2\varepsilon)^{\abs{S}} }^2&\leq \sum_{S\notin \S_1} \hat{\beta}(S)^2\sum_{S\notin \S_1} \hat{\gamma}(\sigma(S))^2 (1-2\varepsilon)^{2\abs{S}}\nonumber\\
    &\leq \sum_{S\notin \S_1} \hat{\beta}(S)^2 < c_t \varepsilon^t,\label{eq:conclusion-of-everything-part1}
\end{align}
and similarly, using \eqref{eq:bourgain-conclusion}, we find
\begin{align}
    \Bigabs{\sum_{S\notin \S_2} \hat{\beta}(S) \hat{\gamma}(\sigma(S)) (1-2\varepsilon)^{\abs{S}} }^2&\leq \frac{1}{100}.\label{eq:conclusion-of-everything-part2}
\end{align}
Thus, by combining \eqref{eq:second-implication},\eqref{eq:conclusion-of-everything-part1}, and \eqref{eq:conclusion-of-everything-part2}, we conclude that \eqref{eq:counclusion-of-everything} holds for sufficiently small $\varepsilon$. 

For our application, we need a slightly refined version of \eqref{eq:counclusion-of-everything}. Since $\beta$ and $\gamma$ are odd functions, we can rewrite \eqref{eq:counclusion-of-everything} as
\begin{align*}
\Big|\sum_{\emptyset \neq S \in \S_1\cap \S_2}\hat{\beta}(S)\hat{\gamma}(\sigma(S))(1-2\varepsilon)^{|S|}\Big| > \frac{1}{4}.
\end{align*}
Squaring both sides and then applying Cauchy-Schwarz yields
\begin{align*}
\frac{1}{16} &< \Big|\sum_{\emptyset \neq S \in \S_1\cap \S_2}\hat{\beta}(S)\hat{\gamma}(\sigma(S))(1-2\varepsilon)^{|S|}\Big|^2\\
& \leq \sum_{\emptyset \neq S \in \S_1\cap \S_2}\hat{\beta}(S)^2 \sum_{\emptyset \neq S \in \S_1\cap \S_2}\hat{\gamma}(\sigma(S))^2(1-2\varepsilon)^{2|S|}\\
& \leq \sum_{\emptyset \neq S \in \S_1\cap \S_2}\hat{\gamma}(\sigma(S))^2(1-2\varepsilon)^{2|S|}
\end{align*}
where in the last step we used Parseval's identity. Since $\frac{1}{4\varepsilon|S|} \geq e^{-4\varepsilon|S|} \geq (1-2\varepsilon)^{2|S|}$, we can rewrite the above as 
\begin{align}\label{eq:intermediate-step}
    \sum_{\emptyset \neq S \in \S_1\cap \S_2}\frac{1}{|S|}\hat{\gamma}(\sigma(S))^2 > \frac{\varepsilon}{4}.
\end{align}
By definition, for all $S \in \S_2$, we have $|\hat{\beta}(S)|^2 > \frac{1}{100}4^{-2\varepsilon^{-2}}$. Combining this with \eqref{eq:intermediate-step}, we get that \[\sum_{\emptyset \neq S \in \S_1\cap \S_2}\frac{1}{|S|}\hat{\beta}(S)^2\hat{\gamma}(\sigma(S))^2\]
can never be smaller than $\frac{\varepsilon}{400}4^{-2\varepsilon^{-2}}$. Thus, we have shown
\begin{align}
    \sum_{\emptyset \neq S \in \S_1\cap \S_2}\frac{1}{|S|}\hat{\beta}(S)^2\hat{\gamma}(\sigma(S))^2 > \frac{\varepsilon}{400}4^{-2\varepsilon^{-2}}.\label{eq:final-equation}
\end{align}

\,

\textbf{Quantum Assignment to $\boldsymbol{\phi}$.} We now construct the quantum assignment $\Pi:U\cup V \to \PVM_m(\hilb)$ to $\phi$ and analyse its value. For every $u \in U\cup V$ and $a \in [m]$, define the operators 
\[P_u^a = \sum_{S:a\in S} \frac{1}{|S|}\hat{\alpha}_u(S)^2,\]
and let $P_u = \{P_u^a\}_{a\in [m]}.$ We may also use the notation $\sum_{S \ni a}$ to mean $\sum_{S:a\in S}$. 

By Lemma \ref{lem:povm_from_obs}, the $P_u$ are self-commuting POVMs. Furthermore, $P_u$ and $P_v$ are simultaneously measurable when $(u,v)\in E$. So by the Projectivization Lemma in Section \ref{sec:projectivization}, there exists a PVM assignment $\Pi$ attaining the same value as $P$. So we just need to calculate the value of $P$.

\,

\textbf{Value of the Assignment $\boldsymbol{P}$.} Substituting into \eqref{eq:2lin-reduction-value-of-ulc}, the value of this assignment is
\begin{align*}
    \omega(\phi,P) = \expect_{u,v} \Brac{ \sum_a \tr\Bigparen{P_u^a P_v^{\pi_{u,v}(a)}}} &= \expect_{u,v} \Brac{ \sum_a \sum_{S\ni a} \, \sum_{T\ni \pi_{u,v}(a)} \frac{1}{\abs{S}\abs{T}}\tr\Bigparen{\hat{\alpha}_u(S)^2 \hat{\alpha}_v(T)^2}}.
\end{align*}
Our first goal is to eliminate the expectation. At the beginning of this section, we established that with probability at least half, an edge $(u,v)\in E$ is a good edge. Among all good edges, let $(\bu,\bv)$ be the one that minimizes the expression inside the expectation:
\[\sum_a \sum_{S\ni a} \, \sum_{T\ni \pi_{u,v}(a)} \frac{1}{\abs{S}\abs{T}}\tr\Bigparen{\hat{\alpha}_u(S)^2 \hat{\alpha}_v(T)^2}.\]
Thus it follows that
\begin{align*}
    \omega(\phi,P) \geq \frac{1}{2} \sum_a \sum_{S\ni a} \, \sum_{T\ni \pi_{\bu,\bv}(a)} \frac{1}{\abs{S}\abs{T}}\tr\Bigparen{\hat{\alpha}_\bu(S)^2 \hat{\alpha}_\bv(T)^2}.
\end{align*}
Next, we aim to eliminate the trace. With probability at least half, an index $j \in [d]$ is good. Among all good indices, let $\bj$ be the one that minimizes
\[ \sum_a \sum_{S\ni a} \, \sum_{T\ni \pi_{\bu,\bv}(a)} \frac{1}{\abs{S}\abs{T}}\hat{\alpha}_{\bu,j}(S)^2 \hat{\alpha}_{\bv,j}(T)^2.\]
Thus, we can write
\begin{align*}
    \omega(\phi,P) \geq \frac{1}{4} \sum_a \sum_{S\ni a} \, \sum_{T\ni \pi_{\bu,\bv}(a)} \frac{1}{\abs{S}\abs{T}}\hat{\alpha}_{\bu,\bj}(S)^2 \hat{\alpha}_{\bv,\bj}(T)^2.
\end{align*}
Ignoring the terms where $T \neq \pi_{\bu,\bv}(S)$, we have
\begin{align*}
    \omega(\phi,P) &\geq \frac{1}{4} \sum_a \sum_{S\ni a} \frac{1}{\abs{S}^2}\hat{\alpha}_{\bu,\bj}(S)^2 \hat{\alpha}_{\bv,\bj}(\pi_{\bu,\bv}(S))^2\\
                &= \frac{1}{4} \sum_{S\neq \emptyset} \frac{1}{\abs{S}}\hat{\alpha}_{\bu,\bj}(S)^2 \hat{\alpha}_{\bv,\bj}(\pi_{\bu,\bv}(S))^2.
\end{align*}
Finally, using \eqref{eq:final-equation}, for every good edge and index, we know that 
\[\sum_{S\neq \emptyset} \frac{1}{\abs{S}}\hat{\alpha}_{\bu,\bj}(S)^2 \hat{\alpha}_{\bv,\bj}(\pi_{\bu,\bv}(S))^2 \geq \frac{\varepsilon}{400}4^{-2\varepsilon^{-2}}.\]
Therefore we conclude that $\omega(\phi,P) \geq \frac{\varepsilon}{1600}4^{-2\varepsilon^{-2}}$, as desired.

\section{MaxCut}\label{sec:maxcut}
In this section, we follow the conventions established in Section \ref{sec:2lin}.
\subsection{Statements}

In this section, we present a proof of the following theorem:

\begin{theorem}\label{thm:maxcuthardness}
    Assuming wqUGC, for every $-1<\rho<0$ and $\varepsilon>0$, $\cut_q(\frac{1-\rho}{2}- \varepsilon, \frac{\arccos \rho}{\pi}+\varepsilon)$ is $\RE$-hard. 
\end{theorem}

This is comparable to Khot et al. \cite{kkmo}:

\begin{theorem}
    Assuming UGC, for every $-1<\rho<0$ and $\varepsilon>0$, $\cut_c(\frac{1-\rho}{2}- \varepsilon, \frac{\arccos \rho}{\pi}+\varepsilon)$ is $\NP$-hard.
\end{theorem}

Our proof is an extension of Khot et al.'s proof to the operator setting.

Theorem \ref{thm:maxcuthardness} follows directly as a corollary of the next theorem:

\begin{theorem}\label{thm:maxcutreduction}
    For every $-1<\rho<0$ and $\varepsilon>0$, there exists $\zeta,\delta>0$ such that
    \begin{align*}
        \ULC_{\wq}(1-\zeta,\delta) \longrightarrow \cut_q\Paren{\frac{1-\rho}{2}-\varepsilon, \frac{\arccos \rho}{\pi} + \varepsilon}.
    \end{align*}
\end{theorem}

To prove Theorem \ref{thm:maxcutreduction}, we will first describe the reduction and then establish its completeness and soundness in the following sections. 

\,

\noindent\textbf{Reduction.}
Fix $-1<\rho<0$ and $\varepsilon>0$. Let $\zeta,\delta>0$ be constants to be determined later. Given an instance $\phi=((U,V),E, \pi, p)$ in $\ULC$, we now describe how to construct an instance $\psi=(\oV,\oE,\opp)$ in $\cut$. Let $m$ be the alphabet size of $\phi$. 

    \,
    
\textbf{Distribution.} We first define a probability distribution. Let $\Pr: U \times V \times V \times \{\pm 1\}^m \times \{\pm 1\}^m \to [0,1]$ be given by \[\Pr(u,v,v',x,\mu) = \frac{p_{u,v} p_{u,v'} }{2^m p_u}\Paren{\frac{1-\rho}{2}}^{|\mu|}\Paren{\frac{1+\rho}{2}}^{m-|\mu|}\] where $p_u$ denotes the marginal distribution $p_u= \sum_{v: (u,v) \in E} p_{u,v}$. This clearly defines a valid distribution. Indeed $\Pr(u,v,v',x,\mu)$ is the probability of sampling
    \begin{itemize}
        \item $u$ from the marginal distribution $p_u$, 
        \item $v$ and $v'$ from the conditional distributions $\frac{p_{u,v}}{p_u}$ and $\frac{p_{u,v'}}{p_u}$, respectively,
        \item $x$ from the uniform distribution on $\{\pm1\}^m$, and 
        \item $\mu$ from the noise distribution $\Paren{\frac{1-\rho}{2}}^{|\mu|}\Paren{\frac{1+\rho}{2}}^{m-|\mu|}$.
    \end{itemize}
We also use $\Pr$ to denote the marginal distributions when no ambiguity arises. For example, for every $(u,v)\in E$, we write $\Pr(u,v)$ to mean $$\sum_{v',x,\mu} \Pr(u,v,v',x,\mu)$$ which is $p_{u,v}$. Similarly, the conditional probability $\Pr(v \lvert u)$ is $$\frac{1}{p_u}\sum_{v',x,\mu}\Pr(u,v,v',x,\mu) = \frac{p_{u,v}}{p_u}.$$

Lastly, whenever we write $\mathbb{E}_{u, v, v', x, \mu}$ (or with any other subset of the indices $u, v, v', x, \mu$), the expectation is taken with respect to the distribution $\Pr$. For example, given a weak-quantum assignment $\Pi:U\cup V \to \Obs(\mathcal{H})$ to $\phi$, its value is given by the following expectation
    \begin{equation}
        \omega(\phi,\Pi) = \sum_{(u,v) \in E} p_{u,v} \sum_{a\in [m]} \tr\Bigparen{\Pi_u^a \Pi_v^{\pi_{u,v}(a)}} = \expect_{u,v}\Brac{\sum_{a\in [m]} \tr\Bigparen{\Pi_u^a \Pi_v^{\pi_{u,v}(a)}}}.\label{eq:value-of-phi}
    \end{equation}

    \,

\textbf{Construction of $\boldsymbol{\psi}$.} We now construct $\psi = (\oV,\oE,\opp)$. The vertex set is defined as $\oV = V \times \{\pm 1\}^m$.  Let $N_u\subseteq V$ denote the set of neighbors of vertex $u\in U$ in the graph $((U,V),E)$. The constraints $\oE$ consists of $((v,x\circ \pi_{v,u}), (v',(x\circ \pi_{v',u}) \mu))$ for all $u\in U$, $v,v'\in N_u$ and $x,\mu \in \{\pm 1\}^m$. The corresponding edge weight is given by $\Pr(u,v,v',x,\mu).$

    \,
    
\textbf{Formula for the Value of $\boldsymbol{\psi.}$} Let $\alpha: \oV \to \Obs(\hilb)$ be a quantum assignment to $\psi$. Using the general formula from Definition \ref{def:quantum-observable-assignment} and the fact that for $\cut$ the parity bits are all equal to $-1$ we obtain $$\omega(\psi,\alpha) = \frac{1}{2} - \frac{1}{2} \sum_{e=(w_1,w_2) \in \oE}\opp_e \tr\Bigparen{\alpha(w_1)\alpha(w_2)}.$$
Writing down the edges and edge weights explicitly, we obtain
\begin{align*}
    \sum_{e=(w_1,w_2) \in \oE}\opp_e \tr\Bigparen{\alpha(w_1)\alpha(w_2)} = \sum_{u,v,v',x,\mu} \Pr(u,v,v',x,\mu) \tr \Bigparen{\alpha_v(x\circ \pi_{v,u}) \alpha_{v'}((x\circ \pi_{v',u}) \mu)}
\end{align*}
which can be rewritten as
\begin{align*}
    \expect_{u,v,v',x,\mu} \Brac{ \tr \Bigparen{\alpha_v(x\circ \pi_{v,u}) \alpha_{v'}((x\circ \pi_{v',u}) \mu)}}.
\end{align*}
Therefore, the value of the assignment is 
\begin{align}\label{eq:maxcut_reduction_value}
    \omega(\psi,\alpha) = \frac{1}{2} - \frac{1}{2} \expect_{u,v,v',x,\mu} \Brac{ \tr \Bigparen{\alpha_v(x\circ \pi_{v,u}) \alpha_{v'}((x\circ \pi_{v',u}) \mu)}}.
\end{align}

\subsection{Completeness}
We first need a technical lemma. 

\begin{lemma}\label{lem:maxcut_completeness_technical}
    Let $\Pi_1, \Pi_2,\Pi_3$ be $m$-outcome PVMs. Then, we have
    \begin{align*}
        \sum_{a\in [m]} \tr\bigparen{\Pi_1^a \Pi_2^a} \geq 2 \sum_a \tr\Bigparen{\Pi_1^a\Pi_3^a + \Pi_2^a \Pi_3^a}-3.
    \end{align*}
\end{lemma}

\begin{proof}
    We calculate
    \begin{align*}
        2-2\sum_a \tr\bigparen{\Pi_1^a \Pi_2^a} &= \sum_a \tr\Bigparen{\Pi_1^a + \Pi_2^a - \Pi_1^a \Pi_2^a - \Pi_2^a \Pi_1^a}\\
        &= \sum_a \Big\| \Pi_1^a - \Pi_2^a\Big\|_2^2\\
        &\leq \sum_a \Bigparen{\Big\|\Pi_1^a - \Pi_3^a\Big\|_2 + \Big\|\Pi_3^a - \Pi_2^a\Big\|_2}^2\\
        &\leq \sum_a 2\Big\|\Pi_1^a - \Pi_3^a\Big\|_2^2 + 2\Big\|\Pi_3^a - \Pi_2^a\Big\|_2^2\\
        &= 8 - 4\sum_a \tr\bigparen{\Pi_1^a \Pi_3^a} - 4\sum_a \tr\bigparen{\Pi_2^a \Pi_3^a}.
    \end{align*}
\end{proof}

Let $\Pi:U\cup V \to \PVM_m(\mathcal{H})$ be a weak-quantum assignment to $\phi$ such that $\omega(\phi,\Pi)\geq 1-\zeta$. We construct a quantum assignment $\alpha: \oV \to \Obs(\mathcal{H})$ to $\psi$ such that $\omega(\psi,\alpha) \geq \frac{1}{2}(1-\rho)(1-4\zeta)$. Choosing $\zeta$ such that $4\zeta<\varepsilon$, we can ensure $\omega(\psi,\alpha) \geq \frac{1-\rho}{2} - \varepsilon$. This will complete the proof of completeness.

For every $v\in V$ and $x\in \{\pm 1\}^m$, define
\begin{align*}
    \alpha_v(x) = \sum_{a\in [m]} x_a \Pi_v^a.
\end{align*}
These operators are Hermitian and unitary. 

Furthermore, for every vertex $u\in U$, vertices $v, v'\in N_u$ and $x,y\in \{ \pm 1\}^m$ we have
    \begin{align*}
    \Brac{\alpha_v(x), \alpha_{v'}(y)} &= \sum_{a,b} x_a y_b \Brac{\Pi_v^a,\Pi_{v'}^b}=0,
    \end{align*}
since $\Pi$ is a weak-quantum assignment, and thus measurements $\Pi_v$ and $\Pi_{v'}$ are simultaneously measurable when $v,v'\in V$ share a neighbor in $U$. We conclude that $\alpha: \oV \to \Obs(\mathcal{H})$ is a quantum assignment to $\psi$. 

We now calculate $\omega(\psi,\alpha)$ using Equation \eqref{eq:maxcut_reduction_value}. We have
\begin{align*}
    \omega(\psi,\alpha) &= \frac{1}{2} - \frac{1}{2} \expect_{u,v,v',x,\mu} \Brac{ \tr\Bigparen{\alpha_v(x\circ \pi_{v,u}) \alpha_{v'}( (x\circ \pi_{v',u}) \mu)} }\\ &= \frac{1}{2} - \frac{1}{2} \expect_{u,v,v',x,\mu} \Brac{\sum_{a,b} x_{\pi_{v,u}(a)} x_{\pi_{v',u}(b)} \mu_b \tr\Bigparen{\Pi_v^a \Pi_{v'}^b}} \\
    &= \frac{1}{2} - \frac{\rho}{2} \expect_{u,v,v'} \Brac{\sum_{c}  \tr\Bigparen{\Pi_v^{\pi_{u,v}(c)} \Pi_{v'}^{\pi_{u,v'}(c)}}},
\end{align*}
where in the third line we used $\expect_x \Brac{x_a x_b} =\delta_{a,b}$, $\expect_\mu \Brac{\mu_b} = \rho$ and we redefined the running index of the sum by $c=\pi_{v,u}(a)=\pi_{v',u}(b)$. Using Lemma \ref{lem:maxcut_completeness_technical}, we have
\begin{align*}
    \expect_{u,v,v'} \Brac{\sum_{c}  \tr\Bigparen{\Pi_v^{\pi_{u,v}(c)} \Pi_{v'}^{\pi_{u,v'}(c)}}} &\geq \expect_{u,v,v'}\Brac{  2 \sum_c \tr\Bigparen{\Pi_u^c \Pi_v^{\pi_{u,v}(c)}} + 2 \sum_c \tr\Bigparen{\Pi_u^c \Pi_{v'}^{\pi_{u,v'}(c)}}-3}\\
    &\geq 4(1-\zeta)-3 = 1-4\zeta,
\end{align*}
where we used the assumption $\omega(\phi,\Pi)\geq 1-\zeta$.
Putting everything together, we obtain
\begin{align*}
    \omega(\psi,\alpha) \geq \frac{1}{2}-\frac{\rho}{2}(1-4\zeta) \geq \frac{1}{2}(1-\rho)(1-4\zeta),
\end{align*}
as desired.

\subsection{Soundness}
Let $\alpha: \oV \to \Obs(\mathcal{H})$ be a quantum assignment to $\psi$ such that $\omega(\psi,\alpha) \geq \frac{\arccos \rho}{\pi} + \varepsilon$. We construct a weak-quantum assignment $\Pi:U\cup V \to \PVM_m(\mathcal{H})$ to $\phi$ such that $\omega(\phi,\Pi) \geq \Omega(\varepsilon^2 \delta_2^3)$ where $\delta_2 = \delta_2(\varepsilon,\rho) > 0$ is a constant to be determined later. This will complete the proof of soundness, as we can choose $\delta$ smaller than $\Omega(\varepsilon^2\delta_2^3)$.

\,

\textbf{Good Vertices in $\boldsymbol{U}$.} By the assumption on the value of $\alpha$ and Equation \eqref{eq:maxcut_reduction_value}, we have
\begin{align*}
    \expect_{u,v,v',x,\mu} \Brac{ \tr \Bigparen{\alpha_v(x\circ \pi_{v,u}) \alpha_{v'}((x\circ \pi_{v',u}) \mu)}} \leq 1-\frac{2}{\pi}\arccos{\rho} - 2\varepsilon.
\end{align*}
Therefore, by a simple averaging argument, with probability at least $\frac{\varepsilon}{2}$ over the vertices $u\in U$, we have
\begin{align}\label{eq:good-u}
    \expect_{v,v',x,\mu} \Brac{ \tr \Bigparen{\alpha_v(x\circ \pi_{v,u}) \alpha_{v'}((x\circ \pi_{v',u}) \mu)}} \leq 1-\frac{2}{\pi}\arccos{\rho} - \varepsilon.
\end{align}
We refer to vertices $u \in U$ satisfying \eqref{eq:good-u} as good. 

Fix a good vertex $\bu$. For every $x\in \{\pm 1\}^m$, define the operator $\beta_\bu(x) = \expect_{v\vert \bu} \Brac{\alpha_v(x\circ \pi_{v,\bu})}$. Rewriting \eqref{eq:good-u} in terms of $\beta$'s, we get
\begin{align}\label{eq:good-u-in-terms-of-beta}
    \expect_{x,\mu} \Brac{ \tr \Bigparen{\beta_\bu(x) \beta_{\bu}(x \mu)}} \leq 1-\frac{2}{\pi}\arccos{\rho} - \varepsilon.
\end{align}

\,

\textbf{Good Classical Assignments.} Due to the commutation relations in the quantum assignment $\alpha$, the operators $\alpha_v(\{\pm1\}^m)$, for all $v \in N_\bu$, are diagonalizable in the same basis. Fix such a basis and express these operators in that basis. For each $j \in [d]$, let $\alpha_{\bu,v,j}(x)$ be the $j$th diagonal entry of $\alpha_v(x)$. For simplicity, we refer to $\alpha_{\bu,v,j}$ as $\alpha_{v,j}$ when the choice of $\bu$ is clear from the context. Similarly, the operators $\beta_\bu(\{\pm1\}^m)$ are also diagonal in this basis. Let $\beta_{\bu,j}(x)$ be the $j$th diagonal entry of $\beta_\bu(x)$.

Note that the range of $\alpha_{v,j}$ is $\{\pm 1\}$. The range of $\beta_{\bu,j}$ on the other hand is the interval $[-1,1]$. 

By the definition of the dimension-normalized trace, we have
\begin{align*}
    \tr \Bigparen{\beta_\bu(x) \beta_{\bu}(x \mu)}= \expect_{j} \Bigbrac{\beta_{\bu,j}(x) \beta_{\bu,j}(x \mu)},
\end{align*}
where the expectation is with respect to the uniform distribution over $[d]$. Thus \eqref{eq:good-u-in-terms-of-beta} becomes
\begin{align}
     \expect_{j} \Brac{\expect_{x,\mu} \Brac{ \beta_{\bu,j}(x) \beta_{\bu,j}(x\mu)}} \leq 1-\frac{2}{\pi}\arccos{\rho} - \varepsilon.
\end{align}
Again, it follows from a simple averaging argument that for at least a $\frac{\varepsilon}{4}$-fraction of the indices $j\in [d]$, we have
\begin{align}\label{eq:maxcut_noise_stab}
    \expect_{x,\mu} \Brac{ \beta_{\bu,j}(x) \beta_{\bu,j}(x\mu)} \leq 1-\frac{2}{\pi}\arccos{\rho} - \frac{\varepsilon}{2}.
\end{align}
We call such indices good.

\,

\textbf{Consequences of the Inequality.} Fix a good index $\bj \in [d]$. We now derive the consequences of \eqref{eq:maxcut_noise_stab}: applying the MIS theorem (Theorem \ref{thm:MIS}) to $\beta_{\bu,\bj}$, we conclude that there exist constants $\delta_2=\delta_2(\frac{\varepsilon}{2},\rho)>0$ and $k=k(\frac{\varepsilon}{2},\rho)$ and a label $\bc \in [m]$ such that the $k$-degree influence of $\bc$ on $\beta_{\bu,\bj}$ is large, that is 
\begin{align}\label{eq:conseq_MIS}
    \Infl_\bc^{\leq k}(\beta_{\bu,\bj}) > \delta_2.
\end{align}

 \,

\textbf{Good Neighbors.} Expanding \eqref{eq:conseq_MIS} in terms of Fourier coefficients
\begin{align*}
    \delta_2&<\Infl_\bc^{\leq k}(\beta_{\bu,\bj}) = \sum_{S: \bc \in S, \abs{S}\leq k} \hat{\beta}_{\bu,\bj}(S)^2 = \sum_{S: \bc \in S, \abs{S}\leq k} \bigparen{\expect_{v\lvert \bu} \Brac{ \hat{\alpha}_{v,\bj}(\pi_{\bu,v}(S))}}^2\\ &\leq \expect_{v\lvert \bu} \Brac{ \sum_{S: \bc \in S, \abs{S}\leq k} \hat{\alpha}_{v,\bj}(\pi_{\bu,v}(S))^2} = \expect_{v\lvert \bu} \Brac{\Infl^{\leq k}_{\pi_{\bu,v}(\bc)}(\alpha_{v,\bj})},
\end{align*}
where we used Jensen's inequality. Again, by a simple averaging argument, with probability at least $\frac{\delta_2}{2}$, a vertex $v \in N_\bu$ satisfies
\begin{align}\label{eq:maxcut_goodv_guarantee}
    \Infl^{\leq k}_{\pi_{\bu,v}(\bc)}(\alpha_{v,\bj})> \frac{\delta_2}{2}.
\end{align} 
We call such $v$ good. 

\,

\textbf{Quantum assignment to $\boldsymbol{\phi}$.}
We now construct the weak-quantum assignment $\Pi: U\cup V \to \Obs(\hilb)$ to $\phi$ and analyse its value. For every $u\in U$, $v\in V$ and $a,b\in [m]$, define the operators
\begin{align*}
    P_u^a = \sum_{S: a \in S} \frac{1}{\abs{S}} \hat{\beta}_u(S)^2 \quad \text{ and } \quad P_v^b = \sum_{S: b\in S} \frac{1}{\abs{S}} \hat{\alpha}_v(S)^2.
\end{align*}
Let $Q_u^m = I -\sum_{a=1}^{m-1} P_u^a$ and $Q_v^m = I -\sum_{b=1}^{m-1} P_v^b$. 
Now let $P_u$ denote the set of operators $P_u^1,\ldots,P_u^m + Q_u^m$. Similarly for $P_v$. We show that $P_u$ and $P_v$ are self-commuting POVMs.

By Lemma \ref{lem:povm_from_obs}, $P_v$'s are self-commuting POVMs. Since $\beta_u(x)$ are not observables, we cannot directly apply Lemma \ref{lem:povm_from_obs} in the case of $P_u$. However, a similar argument as in the proof of Lemma \ref{lem:povm_from_obs} shows that $P_u$ are also $m$-outcome self-commuting POVMs. First observe that 
\begin{align*}
    \sum_a P_u^a = \sum_a \sum_{S: a\in S} \frac{1}{\abs{S}} \hat{\beta}_u(S)^2 = \sum_{S: S\neq \emptyset}\hat{\beta}_u(S)^2
\end{align*}
and expanding in terms of $\alpha_v$'s, we get 
\begin{align*}
    \sum_{S: S\neq \emptyset}\hat{\beta}_u(S)^2 = \sum_{S: S\neq \emptyset} \bigparen{ \expect_v \Brac{\hat{\alpha}_v(\pi_{u,v}(S))}}^2 \leq \expect_v \Brac{ \sum_{S: S\neq \emptyset} \hat{\alpha}_v(\pi_{u,v}(S))^2}  \leq I,
\end{align*}
where we used Jensen's inequality and Parseval's identity. Therefore $Q_u^m$ is a PSD operator and $P_u$ is a POVM. 

POVMs $P_u$ are self-commuting because ultimately $\alpha_v(\{\pm1\}^m)$ are commuting operators for all $v \in N_u$. This also implies that $P_u$ and $P_v$ are simultaneously measurable when $v \in N_u$. Similarly $P_v$ and $P_{v'}$ are simultaneously measurable when $v,v' \in N_u$. Thus $P$ is a weak-quantum POVM assignment to $\phi$. So by the Projectivization Lemma in Section \ref{sec:projectivization}, there exists a weak-quantum PVM assignment $\Pi$ attaining the same value as $P$. So we just need to calculate the value of $P$.

\,

\textbf{Value of the Assignment $\boldsymbol{P}$.} By \eqref{eq:value-of-phi}, the value of the assignment satisfies
\begin{align*}
    \omega(\phi,P) &\geq \expect_{u,v} \Brac{ \sum_a \tr\Bigparen{P_u^a P_v^{\pi_{u,v}(a)}}}
\end{align*}

where the inequality is the result of ignoring the contributions of $Q_u^m$ and $Q_v^m$. Expanding in terms of Fourier coefficients
\[\omega(\phi,P) \geq \expect_{u,v} \Brac{ \sum_a \sum_{S\ni a} \, \sum_{T\ni \pi_{u,v}(a)} \frac{1}{\abs{S}\abs{T}}\tr\Bigparen{\hat{\beta}_u(S)^2 \hat{\alpha}_v(T)^2}}.\]

Our first goal is to eliminate the expectation over $u$. We have established that with probability at least $\frac{\varepsilon}{2}$, a vertex $u\in U$ is a good vertex. Among all good vertices, let $\bu$ be the one that minimizes the expression inside the expectation over $u$:
\begin{align*}
    \expect_{v\lvert u} \Brac{ \sum_a \sum_{S\ni a} \, \sum_{T\ni \pi_{u,v}(a)} \frac{1}{\abs{S}\abs{T}}\tr\Bigparen{\hat{\beta}_u(S)^2 \hat{\alpha}_v(T)^2}}.
\end{align*}
Thus it follows that
\begin{align*}
    \omega(\phi,P) \geq \frac{\varepsilon}{2} \expect_{v\lvert \bu} \Brac{ \sum_a \sum_{S\ni a} \, \sum_{T\ni \pi_{\bu,v}(a)} \frac{1}{\abs{S}\abs{T}}\tr\Bigparen{\hat{\beta}_\bu(S)^2 \hat{\alpha}_v(T)^2}}.
\end{align*}
Next, we aim to eliminate the trace. With probability at least $\frac{\varepsilon}{4}$, an index $j\in [d]$ is good. Among all good indices, let $\bj$ be the one that minimizes
\begin{align*}
    \expect_{v\lvert \bu} \Brac{\sum_a \sum_{S\ni a} \, \sum_{T\ni \pi_{\bu,v}(a)} \frac{1}{\abs{S}\abs{T}}\hat{\beta}_{\bu,j}(S)^2 \hat{\alpha}_{v,j}(T)^2}.
\end{align*}
Thus we have
\begin{align}\label{eq:maxcut_intermediate}
    \omega(\phi,P) \geq \frac{\varepsilon^2 }{8} \expect_{v\lvert \bu} \Brac{\sum_a \sum_{S\ni a} \, \sum_{T\ni \pi_{\bu,v}(a)} \frac{1}{\abs{S}\abs{T}}\hat{\beta}_{\bu,\bj}(S)^2 \hat{\alpha}_{v,\bj}(T)^2}.
\end{align}
We saw that using MIS, for fixed $\bu$ and $\bj$, we get a label $\bc \in [m]$, for which \eqref{eq:conseq_MIS} holds. Also with probability at least $\frac{\delta_2}{2}$, we get a neighbor $\bv$ of $\bu$ that satisfies \eqref{eq:maxcut_goodv_guarantee}. Applying a similar argument to the one we have been using to eliminate expectations, we have
\begin{align*}
    \omega(\phi,P) \geq \frac{\varepsilon^2 \delta_2}{16} \sum_a \sum_{S\ni a} \, \sum_{T\ni \pi_{\bu,\bv}(a)} \frac{1}{\abs{S}\abs{T}}\hat{\beta}_{\bu,\bj}(S)^2 \hat{\alpha}_{\bv,\bj}(T)^2,
\end{align*}
for some good $\bv$. Ignoring the terms where $a$ is not equal to the label $\bc$, we have 
\begin{align*}
    \omega(\phi,P) \geq \frac{\varepsilon^2 \delta_2}{16} \sum_{S\ni \bc} \, \sum_{T\ni \pi_{\bu,\bv}(\bc)} \frac{1}{\abs{S}\abs{T}}\hat{\beta}_{\bu,\bj}(S)^2 \hat{\alpha}_{\bv,\bj}(T)^2.
\end{align*}
Ignoring the terms where $\abs{S}, \abs{T}> k$, we have
\begin{align*}
    \omega(\phi,P) &\geq \frac{\varepsilon^2 \delta_2}{16k^2} \sum_{S\ni \bc, \abs{S}\leq k} \,  \hat{\beta}_{\bu,\bj}(S)^2 \sum_{T\ni \pi_{\bu,\bv}(\bc), \abs{T}\leq k} \hat{\alpha}_{\bv,\bj}(T)^2.
\end{align*}
Finally, using \eqref{eq:conseq_MIS} and \eqref{eq:maxcut_goodv_guarantee}, we have
\begin{align*}
    \omega(\phi,P) \geq \frac{\varepsilon^2 \delta_2}{16k^2}\, \Infl_\bc^{\leq k}(\beta_{\bu,\bj})\, \Infl^{\leq k}_{\pi_{\bu,\bv}(\bc)}(\alpha_{\bv,\bj}) \geq \frac{\varepsilon^2 \delta_2^3}{32k^2},
\end{align*}
as desired.

\subsection{Integrality Gap}
In Section \ref{sec:result-on-quantum-maxcut} and Figure \ref{fig:extremes}, we suggested that, assuming wqUGC, there are instances $G$ of MaxCut for which $\omega_c(G) \approx 0.878\,\omega_{\mathrm{sdp}}(G)$ and $\omega_q(G) \approx \omega_{\mathrm{sdp}}(G)$. This is a direct consequence of the following theorem:
\begin{theorem}\label{thm:integrality-gap}
Assuming wqUGC, there exists an instance $G$ of MaxCut for which $\omega_c(G) = \alpha_{\mathrm{gw}}\,\omega_q(G)$, where $\alpha_{\mathrm{gw}} = \frac{2}{\pi} \min_{0\leq \theta\leq \pi}\frac{\theta}{1-\cos(\theta)}$ is the Goemans-Williamson constant.
\end{theorem}
\begin{proof}
First consider the chain of inequalities $\alpha_{\mathrm{gw}}\,\omega_q(G) \leq \omega_c(G) \leq \omega_q(G)$ derived from \eqref{eq:chain-of-inequalities}. The following three straightforward facts show that the inequality on the left, $\alpha_{\mathrm{gw}}\, \omega_q(G) \leq \omega_c(G)$, must be tight:
\begin{enumerate}
    \item the classical value can be exactly computed in $\NP$,
    \item it is $\RE$-hard to approximate the quantum value better than $\alpha_{\mathrm{gw}}$, and
    \item $\RE \neq \NP$.
\end{enumerate}
So for every $\varepsilon>0$, there must exist an instance $G$ for which $\omega_c(G) \leq (\alpha_{\mathrm{gw}}+ \varepsilon)\, \omega_q(G) $.
\end{proof}
In Section \ref{sec:future}, we raised the question of whether such an instance $G$ can be explicitly constructed. Compare this question to the celebrated integrality gap paper of Feige and Schechtman~\cite{feige_integrality}, where they constructed an example showing the tightness of the inequality $\alpha_{\mathrm{gw}}\,\omega_{\mathrm{sdp}}(G)\leq \omega_c(G)$.

\section{Projectivization Lemma}\label{sec:projectivization}
A (weak-quantum, quantum, or noncommutative) \emph{POVM assignment} is defined similarly to a (weak-quantum, quantum, or noncommutative) assignment, except that we use self-commuting POVMs instead of PVMs in the definitions. In this section, we aim to show that the choice between PVMs and self-commuting POVMs does not affect the value. Specifically, we prove that if a POVM assignment achieves a value $\lambda$, then there also exists a PVM assignment that achieves a value of at least $\lambda$.

\textbf{Extremal POVMs.} The set of POVMs is convex and compact, so by the Krein-Milman theorem, every POVM can be expressed as a convex combination of extremal POVMs. According to Holevo~\cite{holevo}, PVMs are extremal POVMs. However it is important to note that not all extremal POVMs are PVMs. For more details, see Watrous~\cite{watrous-book}. 

The set of self-commuting POVMs is also convex and compact, with PVMs remaining extremal in this smaller subset. In fact, PVMs are the only extremal points in the set of self-commuting POVMs. Thus, we can express any self-commuting POVM as convex combinations of some PVMs. However, to achieve the goal of this section, we must proceed carefully: the POVMs in a POVM assignment satisfy certain commutation relations, which we must account for. 

\textbf{Objective Function.} Let $P = \{P_1,\ldots,P_n\} \subset \POVM_m(\hilb)$ be a quantum POVM assignment to a $\CSP{k}(m)$ instance $\phi = ([n],E,f,p)$. Recall that $\POVM_m$ denotes the set of self-commuting POVMs. The value of this assignment is given by
$$\omega(\phi,P) = \sum_{e=(i_1,\ldots,i_k)\in E} p_e\sum_{a \in [m]^k} f_e(a_1,\ldots,a_k)\tr\Paren{P_{i_1}^{a_1}\cdots P_{i_k}^{a_k}}.$$

Let us write $\omega_\phi(P_1,\ldots,P_n)$ to represent $\omega(\phi,P)$. We view $\omega_\phi$ as a function $\POVM_m^n \to \C$. On a quantum assignment, due to the commutation relations, the image is guaranteed to lie in the interval $[0,1]$. However, since the set of quantum assignments is not convex, we instead consider the domain to be the convex and compact set $\POVM_m^n$. 

This function behaves well coordinate-wise: for every $P_1,\ldots,P_n,Q \in \POVM_m(\hilb)$, $\lambda \in [0,1]$, and $i \in [n]$, we have the property
\begin{align}\label{eq:property-projectivization}
    \omega_\phi(P_1,\ldots,\lambda P_i + (1-\lambda)Q,\ldots,P_n) = \lambda\omega_\phi(P_1,\ldots,P_i,\ldots,P_n) + (1-\lambda)\omega_\phi(P_1,\ldots,Q,\ldots,P_n).
\end{align}
This property holds because each POVM $P_i$ contributes linearly to $\omega_\phi(P_1,\ldots,P_n)$. Specifically, at most one element $P_i^a \in P_i$, for some $a \in [m]$, appears in each term of the summation. 

Given \eqref{eq:property-projectivization} and the convexity of the domain, the function is coordinate-wise both convex and concave. However, strictly speaking, convexity is defined for real-valued functions, so we will not use this terminology here.

\begin{lemma}\label{lemma:povm-assignment-is-the-same-as-pvm} Let $\phi = (V,E,f,p)$ be a $\CSP{k}(m)$ and let $P:V\to \POVM_m(\hilb)$ be a quantum POVM assignment to $\phi$. There exists a quantum (PVM) assignment $\Pi:V \to \PVM_m(\hilb)$ to $\phi$ for which the value is at least the value of $P$.
\end{lemma}
Note that $\Pi$ is defined on the same Hilbert space as $P$. By Property \eqref{eq:property-projectivization} of the objective function, Lemma \ref{lemma:povm-assignment-is-the-same-as-pvm} is a direct corollary of the following lemma.
\begin{lemma}\label{lem:projlemma}
    For every $P \in \POVM_m(\mathcal{H})$ there exists an integer $N$ and $\Pi_1,\Pi_2,\ldots,\Pi_N \in \PVM_m(\mathcal{H})$ such that $P$ is in the convex hull of $\Pi_i$'s. Furthermore, if $M\in \mathcal{M}(\mathcal{H})$ is any matrix that commutes with $P$, it also commutes with every $\Pi_i$.
\end{lemma}
\begin{proof}
Let $P = \{P^1,\ldots,P^m\}$. Let $\calA$ be the commutative algebra generated by $P^1,\ldots,P^m$. Let $\Delta_1,\ldots,\Delta_N$ be a complete set of minimal projections in $\calA$: that is $\Delta_i$ are minimal projections and $\Delta_i\Delta_j = 0$ if $i\neq j$ and $\Delta_1 + \cdots + \Delta_N = I$. Since $\calA$ is a finite-dimensional commutative algebra, it is straightforward to show that such a complete set of minimal projections always exists. 

Let $\POVM_m(\calA)$ denote the set of $m$-outcome POVMs on $\calA$. For example $P \in \POVM_m(\calA)$. Similarly let $\PVM_m(\calA)$ denote the set of $m$-outcome PVMs on $\calA$. By Holevo's theorem, elements of $\PVM_m(\calA)$ are extremal in the larger set of all POVMs, and therefore, they must also be extremal in the smaller set $\POVM_m(\calA)$. Using a perturbative argument, we show that these are the only extremal measurements in $\POVM_m(\calA)$. 

Suppose $Q = \{Q^1,\ldots,Q^m\}$ is extremal in $\POVM_m(\calA)$. For each $a\in [m]$ we can express $Q^a = \sum_i \alpha_{a,i}\Delta_i$ for some $\alpha_{a,i} \in [0,1]$. If $Q$ is not a PVM, then there must exist some $a\in [m]$ and $i\in [N]$ such that $\alpha_{a,i} \in (0,1)$. Then there must also exist a $b\in [m]\setminus \{a\}$ for which $\alpha_{b,i} \in (0,1)$. For simplicity, assume $a = 1$ and $b = 2$. For sufficiently small $\varepsilon > 0$, consider two POVMs $$Q^1 + \varepsilon \Delta_i,Q^2 - \varepsilon\Delta_i,Q^3,\ldots,Q^m$$ and $$Q^1 - \varepsilon \Delta_i,Q^2 + \varepsilon\Delta_i,Q^3,\ldots,Q^m.$$ Clearly they are in $\calA$ and $Q$ is a convex combination of them. This is a contradiction. 

Since $\POVM_m(\hilb)$ is convex and compact, given $P \in \POVM_m(\hilb)$, by the Krein-Milman theorem, there exist extremal $\Pi_1,\ldots,\Pi_N \in \POVM_m(\hilb)$ such that $P$ is in their convex hull. Since extremal measurements in $\POVM_m(\hilb)$ are PVMs this completes the proof of the first part. 

The second part follows from the fact that every element in $\calA$ is a polynomial in commuting operators $P^1,\ldots,P^m$. So if $M$ is a matrix commuting with $P^1,\ldots,P^m$, it commutes with every element in $\calA$ including $\Pi_i$'s.
\end{proof}

\section{Folding Lemma}\label{sec:folding-lemma}
\begin{lemma*}
Given an instance $\psi$ as constructed in the proof of Theorem \ref{thm:lin2reduction}, there exists another efficiently constructed instance $\psi' \in \Lin{2}$ such that assignments to $\psi'$ one-to-one correspond to folded assignments to $\psi$ with the same value. Consequently, the folded value of $\psi$ is the same as the ordinary value of $\psi'$.
\end{lemma*}
\begin{proof}
    Let $\xi_m$ be a subset of $\{\pm 1\}^m $ such that for every $x\in \{\pm 1\}^m$ exactly one of $x$ or $-x$ is contained in $\xi_m$. For every $x\in \{\pm 1\}^m$, let $\kappa(x)=1$ if $x\in \xi_m$ and $\kappa(x)=-1$ otherwise. Finally let $\theta(x)=\kappa(x) x$. Clearly $\theta(x) \in \xi_m$ for all $x$. 
    
    We now define the instance $\psi' = (\oV',\oE',\orr',\opp')$. Let $\oV' = (U \cup V)\times \xi_m$. The edges $e \in \oE'$, their parity bits $\orr'_e \in \{\pm 1\}$, and their weights $\opp'_e$ are defined as follows:
    \begin{enumerate}
        \item For every $u \in U \cup V$ and $x,\mu \in \{\pm1\}^m$, let $((u,\theta(x)),(u,\theta(x\mu)))$ be an edge in $\oE'$. Set its parity bit to $\kappa(x)\kappa(x\mu)$. Set its weight to be the same as the weight of $((u,x),(u,x\mu))$ in $\psi$.\footnote{Note that $\oE'$ is allowed to be a multiset (see Definition \ref{def:general-csp}). For example, both $((u,\theta(x)),(u,\theta(x\mu)))$ and $((u,\theta(-x)),(u,\theta(-x\mu)))$ are edges in $\oE'$, each with their own parity bits and weights.}
        
        \item For every edge $(u,v) \in E$ and $x,\mu \in \{\pm 1\}^m$, let $((u,\theta(x)),(v,\theta((x\circ \pi_{v,u})\mu)))$ be an edge in $\oE'$. Set its parity bit to $\kappa(x)\kappa((x\circ \pi_{v,u})\mu)$. Set its weight to be the same as the weight of $((u,x),(v,(x\circ \pi_{v,u})\mu))$ in $\psi$.
    \end{enumerate}
    Let $\alpha': \oV' \to \Obs(\mathcal{H})$ be an arbitrary (classical, quantum, or noncommutative) assignment to $\psi'$. Define the assignment $\alpha: \oV\to\Obs(\mathcal{H})$ to $\psi$ such that $\alpha(u,x) = \kappa(x)\alpha'(u,\theta(x))$ for all $u \in U\cup V$ and $x \in \{\pm1\}^m$. It holds that $$\alpha(u,-x) = \kappa(-x)\alpha'(u,\theta(-x)) = -\kappa(x)\alpha'(u,\theta(x)) = -\alpha(u,x)$$
    and thus $\alpha$ is folded. Also given any folded assignment $\alpha$ to $\psi$, the assignment $\alpha'$ is uniquely determined using the identity $\alpha'(u,\theta(x)) = \kappa(x)\alpha(u,x)$. This is well-defined since $\alpha'(u,\theta(-x)) = \kappa(-x)\alpha(u,-x) = \kappa(x)\alpha(u,x) = \alpha'(u,\theta(x))$. This completes the proof of one-to-one correspondence. 
    
    Next we show that $\alpha$ satisfies an edge in $\oE$ if and only if $\alpha'$ satisfies the corresponding edge in $\oE'$. We do this for one of the two types of edges in the constructions of $\psi$ and $\psi'$. The other case can be proved similarly. 
    
    Fix an arbitrary edge $(u,v) \in E$ and $x,\mu \in \{\pm 1\}^m$. By the construction of $\psi$, assignment $\alpha$ satisfies the edge $((u,x),(v,(x\circ \pi_{v,u})\mu))$ if and only if $\alpha_u(x) = \alpha_v((x\circ \pi_{v,u})\mu)$. By the definition of $\alpha$, this is equivalent to $\kappa(x)\alpha'_u(\theta(x)) = \kappa((x\circ \pi_{v,u})\mu)\alpha'_v(\theta((x\circ \pi_{v,u})\mu))$. Finally by the construction of $\psi'$, this is equivalent to saying that $\alpha'$ satisfies the edge $((u,\theta(x)),(v,\theta((x\circ \pi_{v,u})\mu))).$
\end{proof}

\bibliographystyle{unsrt}
\bibliography{refs}
\end{document}

%% file: macros.tex
\usepackage{bold-extra}
\usepackage[nottoc]{tocbibind}
\usepackage[normalem]{ulem}

\usepackage{tocloft}
\usepackage[toc]{multitoc}

\usepackage[affil-it]{authblk}
\usepackage{microtype}
\usepackage[margin=1in]{geometry}
\usepackage{amssymb}
\usepackage{tikz}
\usepackage{hyperref}
\hypersetup{
    colorlinks=true,
    linkcolor=blue,
    filecolor=magenta,      
    urlcolor=cyan,
    }
\usepackage{bigints}
\usepackage{amsthm}
\usepackage{float} 
\usepackage{mathtools}
\usepackage[font=small,labelfont=bf]{caption}
\usepackage{xparse}
\usepackage{subcaption}
\usepackage{suffix}
\usepackage{tikz}

\renewcommand{\hat}[1]{\widehat{#1}} 

\newcommand{\oV}{\overline{V}}
\newcommand{\oE}{\overline{E}}
\newcommand{\orr}{\overline{r}}
\newcommand{\opp}{\overline{p}}

\newcommand{\bu}{\mathbf{u}}
\newcommand{\bv}{\mathbf{v}}
\newcommand{\bj}{\mathbf{j}}
\newcommand{\be}{\mathbf{e}}
\newcommand{\bc}{\mathbf{c}}

\newcommand{\Z}{\mathbb{Z}}
\newcommand{\U}{\mathcal{U}}
\renewcommand{\P}{\mathcal{P}}
\renewcommand{\S}{\mathcal{S}}

\newcommand{\XOR}{3\mathrm{XOR}}

\newcommand{\CSP}[1]{{#1}\textit{-}\mathrm{CSP}}

\newcommand{\LC}{\mathrm{LC}}
\newcommand{\ULC}{\mathrm{ULC}}
\newcommand{\UGC}{\mathrm{UGC}}
\newcommand{\Lin}[1]{{#1}\textit{-}\mathrm{Lin}}
\newcommand{\cut}{\mathrm{MaxCut}}

\newcommand{\wq}{\mathrm{wq}}
\newcommand{\sq}{\mathrm{sq}}
\newcommand{\nc}{\mathrm{nc}}
\newcommand{\POVM}{\mathrm{POVM}}
\newcommand{\PVM}{\mathrm{PVM}}

\newcommand{\calA}{\mathcal{A}}

\numberwithin{equation}{section}
\newtheorem{theorem}{Theorem}

\newtheorem{lemma}[theorem]{Lemma}

\newtheorem*{claim*}{Claim}
\newtheorem{conjecture}[theorem]{Conjecture}

\newtheorem*{conjecture*}{Conjecture}
\newtheorem*{definition*}{Definition}
\newtheorem*{theorem*}{Theorem}
\newtheorem*{remark*}{Remark}
\newtheorem*{lemma*}{Lemma}
\newtheorem*{question*}{Question}
\newtheorem{definition}[theorem]{Definition}

\newtheorem*{example*}{Example}

\newcommand{\R}{\mathbb{R}}
\newcommand{\C}{\mathbb{C}}
\newcommand{\N}{\mathbb{N}}

\newcommand{\tr}{\operatorname{tr}}

\newcommand{\abs}[1]{\lvert #1 \rvert}

\newcommand{\Bigabs}[1]{\Bigl\lvert #1 \Bigr\rvert}

\newcommand{\M}{\mathcal{M}}
\newcommand{\expect}{\mathbb{E}}
\newcommand{\Obs}{\mathrm{Obs}}

\newcommand{\Paren}[1]{\left(#1\right)}
\newcommand{\bigparen}[1]{\big(#1\big)}
\newcommand{\Bigparen}[1]{\Big(#1\Big)}

\newcommand{\NP}{\mathrm{NP}} %
\newcommand{\Pp}{\mathrm{P}} %
\newcommand{\QMA}{\mathrm{QMA}} %
\newcommand{\MIP}{\mathrm{MIP}} %
\newcommand{\RE}{\mathrm{RE}} %

\newcommand{\poly}{\mathrm{poly}}

\newcommand{\hilb}{\cal{H}}

\let\epsilon=\varepsilon %

\newcommand{\Brac}[1]{\left[#1\right]}
\newcommand{\bigbrac}[1]{\big[#1\big]}
\newcommand{\Bigbrac}[1]{\Big[#1\Big]}

\newcommand{\Infl}{\mathrm{Inf}}